\documentclass[preprint,numbers]{sigplanconf}

\usepackage{amsmath}

\usepackage{amsthm}
\usepackage{amssymb}
\usepackage{mathtools}
\usepackage{verbatim}
\usepackage{color}
\usepackage{boxedminipage}
\usepackage{enumitem}
\usepackage[caption=false]{subfig}
\usepackage{url}
\usepackage{hyperref}

\begin{document}

\setlength{\pdfpageheight}{\paperheight}
\setlength{\pdfpagewidth}{\paperwidth}

\conferenceinfo{CONF 'yy}{Month d--d, 20yy, City, ST, Country}
\copyrightyear{20yy}
\copyrightdata{978-1-nnnn-nnnn-n/yy/mm}
\copyrightdoi{nnnnnnn.nnnnnnn}



\newcounter{desccount}
\newcommand{\descitem}[1]{%
	\item[#1] \refstepcounter{desccount}\label{#1}
}
\newcommand{\descref}[1]{\hyperref[#1]{#1}}

\newcommand{\ie}{i.e.}
\newcommand{\etc}{etc.}
\newcommand{\eg}{e.g.}
\newcommand{\etal}{et al.}
\newcommand{\wrt}{w.r.t.}
\mathchardef\mhyphen="2D
\newcommand{\Hyphen}{\mhyphen}

\newtheorem{lemma}{Lemma}
\newtheorem{theorem}{Theorem}
\newtheorem{definition}{Definition}
\newtheorem{observation}{Observation}

\newcommand{\mycomment}[1]{\emph{\textcolor{red}{[#1]}}}

\newcommand{\codecomment}[1]{\textcolor{blue}{#1}}

\newcommand{\NmInst}{\mathsf{Nm}}
\newcommand{\LdInst}{\mathsf{Ld}}
\newcommand{\StInst}{\mathsf{St}}
\newcommand{\AluInst}{\mathsf{AluBr}}
\newcommand{\ComInst}{\mathsf{Commit}}
\newcommand{\RecInst}{\mathsf{Reconcile}}
\newcommand{\RMWInst}{\mathsf{RMW}}
\newcommand{\SWPInst}{\mathsf{XCHG}}
\newcommand{\LLInst}{\mathsf{LdL}}
\newcommand{\SCInst}{\mathsf{StC}}

\newcommand{\syncInst}{\mathsf{sync}}
\newcommand{\lwsyncInst}{\mathsf{lwsync}}
\newcommand{\isyncInst}{\mathsf{isync}}
\newcommand{\cmpInst}{\mathsf{cmp}}
\newcommand{\bcInst}{\mathsf{bc}}

\newcommand{\dmbInst}{\mathsf{dmb}}

\newcommand{\MBInst}{\mathsf{MEMBAR}}

\newcommand{\wmmDep}{WMM-D}
\newcommand{\oooNVP}{OOO-D}
\newcommand{\wmmSSB}{WMM-S}

\newcommand{\noDepProg}{simple}
\newcommand{\oooBase}{OOO-B}
\newcommand{\oooSimple}{OOO-BS}
\newcommand{\oooCoOpt}{OOO-C}
\newcommand{\oooSpec}{OOO-S}
\newcommand{\oooSSB}{OOO-H} 
\newcommand{\oooSSBMid}{OOO-H'} 
\newcommand{\oooSSBSimple}{OOO-HS}
\newcommand{\oooWmmSSB}{OOO-SSB} 

\newcommand{\IIE}{\ensuremath{\mathrm{I^2E}}}

\newcommand{\assignVal}{\mathsf{\ \coloneqq\ }}
\newcommand{\True}{\mathsf{True}}
\newcommand{\False}{\mathsf{False}}
\newcommand{\ifFunc}{\mathbf{if}}
\newcommand{\thenFunc}{\mathbf{then}}
\newcommand{\elseFunc}{\mathbf{else}}
\newcommand{\whenFunc}{\mathsf{when}}
\newcommand{\decodeFunc}{\mathsf{decode}}
\newcommand{\executeFunc}{\mathsf{execute}}
\newcommand{\emptyFunc}{\mathsf{empty}}
\newcommand{\existFunc}{\mathsf{exist}}
\newcommand{\getYoungestFunc}{\mathsf{youngest}}
\newcommand{\enqFunc}{\mathsf{enq}}
\newcommand{\deqFunc}{\mathsf{deq}}
\newcommand{\getAnyAddr}{\mathsf{anyAddr}}
\newcommand{\insertFunc}{\mathsf{insert}}
\newcommand{\removeOldestFunc}{\mathsf{rmOldest}}
\newcommand{\removeOlderFunc}{\mathsf{rmOlder}}
\newcommand{\removeAddrFunc}{\mathsf{rmAddr}}
\newcommand{\getRandAndRemoveFunc}{\mathsf{getAny}}
\newcommand{\clearFunc}{\mathsf{clear}}
\newcommand{\hasStFunc}{\mathsf{hasSt}}
\newcommand{\getRandFunc}{\mathsf{any}}
\newcommand{\getOldestFunc}{\mathsf{oldest}}
\newcommand{\hasTagFunc}{\mathsf{hasTag}}
\newcommand{\maxFunc}{\mathsf{max}}
\newcommand{\letFunc}{\mathbf{let}}
\newcommand{\inFunc}{\mathbf{in}}
\newcommand{\decodeTSFunc}{\mathsf{decodeTS}}
\newcommand{\executeTSFunc}{\mathsf{executeTS}}
\newcommand{\existTSFunc}{\mathsf{existTS}}
\newcommand{\getRandAndRemoveTSFunc}{\mathsf{getRandomTS}}
\newcommand{\noCycleFunc}{\mathsf{noCycle}}

\newcommand{\coOrd}{\xrightarrow{co}}

\newcommand{\scNmRule}{SC-Nm}
\newcommand{\scLdRule}{SC-Ld}
\newcommand{\scStRule}{SC-St}

\newcommand{\tsoNmRule}{TSO-Nm}
\newcommand{\tsoLdRule}{TSO-Ld}
\newcommand{\tsoStRule}{TSO-St}
\newcommand{\tsoComRule}{TSO-Com}
\newcommand{\tsoDeqSbRule}{TSO-DeqSb}

\newcommand{\psoDeqSbRule}{PSO-DeqSb}

\newcommand{\wmmNmRule}{WMM-Nm}
\newcommand{\wmmLdSbRule}{WMM-LdSb}
\newcommand{\wmmLdMemRule}{WMM-LdMem}
\newcommand{\wmmLdIbRule}{WMM-LdIb}
\newcommand{\wmmStRule}{WMM-St}
\newcommand{\wmmRecRule}{WMM-Rec}
\newcommand{\wmmComRule}{WMM-Com}
\newcommand{\wmmDeqSbRule}{WMM-DeqSb}
\newcommand{\wmmRMWRule}{WMM-RMW}

\newcommand{\wmmDepNmRule}{\wmmDep-Nm}
\newcommand{\wmmDepStRule}{\wmmDep-St}
\newcommand{\wmmDepLdSbRule}{\wmmDep-LdSb}
\newcommand{\wmmDepLdIbRule}{\wmmDep-LdIb}
\newcommand{\wmmDepLdMemRule}{\wmmDep-LdMem}
\newcommand{\wmmDepRecRule}{\wmmDep-Rec}
\newcommand{\wmmDepComRule}{\wmmDep-Com}
\newcommand{\wmmDepDeqSbRule}{\wmmDep-DeqSb}

\newcommand{\wmmSSBPropSt}{\wmmSSB-Copy}
\newcommand{\wmmSSBLdRemoteRule}{\wmmSSB-Ld-Remote}
\newcommand{\wmmSSBStRule}{\wmmSSB-St}
\newcommand{\wmmSSBDeqSbRule}{\wmmSSB-DeqSb}

\newcommand{\rcLdBeforeStCons}{RC-LdVal-1}
\newcommand{\rcStBeforeLdCons}{RC-LdVal-2}
\newcommand{\rcAcqOrdCons}{RC-Acquire-Order}
\newcommand{\rcRelOrdCons}{RC-Release-Order}
\newcommand{\rcSpecialOrdCons}{RC-Special-Order}
\newcommand{\rcDepCons}{RC-Dependency}
\newcommand{\rcCoCons}{RC-Coherence}
\newcommand{\rcNoDLCons}{RC-No-Deadlock}
\newcommand{\rcFixLdValCons}{RC-Fix-WC-Ld-Value}
\newcommand{\rcFixDepCons}{RC-Fix-Dependency}
\newcommand{\rcFixStore}{RC-Fix-Store}
\newcommand{\rcFixForward}{RC-Fix-Forward}

\newcommand{\rmoFixLd}{RMO-Fix-Ld}

\newcommand{\coEdge}{\xrightarrow{co}}
\newcommand{\rfEdge}{\xrightarrow{r\!f}}
\newcommand{\poEdge}{\xrightarrow{po}}

\newcommand{\robModelName}{OOO\textsuperscript{VP}}
\newcommand{\robDepModelName}{OOO\textsuperscript{D}}
\newcommand{\robSSBModelName}{OOO\textsuperscript{S}}
\newcommand{\ccmModelName}{CCM}
\newcommand{\flowModelName}{HMB} 
\newcommand{\robccm}{\ccmModelName+\robModelName}
\newcommand{\robDepccm}{\ccmModelName+\robDepModelName}
\newcommand{\robSSBflow}{\flowModelName+\robSSBModelName}

\newcommand{\fm}{hmb}

\newcommand{\rf}{r\!f}
\newcommand{\lf}{l\!f}

\newcommand{\Idle}{\mathsf{Idle}}
\newcommand{\Exe}{\mathsf{Exe}}
\newcommand{\ReEx}{\mathsf{ReEx}}
\newcommand{\Done}{\mathsf{Done}}
\newcommand{\getPCFunc}{\mathsf{getPC}}
\newcommand{\getReadyFunc}{\mathsf{getReady}}
\newcommand{\computNmFunc}{\mathsf{computeNm}}
\newcommand{\computeAddrFunc}{\mathsf{computeAddr}}
\newcommand{\computeStDataFunc}{\mathsf{computeStData}}
\newcommand{\updateFunc}{\mathsf{update}}
\newcommand{\getEntryFunc}{\mathsf{getEntry}}
\newcommand{\getLdFunc}{\mathsf{getLd}}
\newcommand{\flushFunc}{\mathsf{flush}}
\newcommand{\findRecStFunc}{\mathsf{findBypass}}
\newcommand{\findLdKilledByStFunc}{\mathsf{findAffectedLd}}
\newcommand{\findLdKilledByLdFunc}{\mathsf{findStaleLd}}
\newcommand{\predictFunc}{\mathsf{predict}}
\newcommand{\fetchFunc}{\mathsf{fetch}}
\newcommand{\getCommitFunc}{\mathsf{getCommit}}
\newcommand{\setCommitFunc}{\mathsf{setCommit}}
\newcommand{\isNmLdRecFunc}{\mathsf{isNmLdRec}}
\newcommand{\initExFunc}{\mathsf{initEx}}
\newcommand{\getExFunc}{\mathsf{getEx}}
\newcommand{\setReExFunc}{\mathsf{setReEx}}
\newcommand{\issueFunc}{\mathsf{issue}}

\newcommand{\reqLdFunc}{\mathsf{reqLd}}
\newcommand{\reqStFunc}{\mathsf{reqSt}}
\newcommand{\reqComFunc}{\mathsf{reqCom}}
\newcommand{\respLdFunc}{\mathsf{respLd}}
\newcommand{\respStFunc}{\mathsf{respSt}}
\newcommand{\respComFunc}{\mathsf{respCom}}
\newcommand{\removeFunc}{\mathsf{remove}}

\newcommand{\robFetchRule}{OOO-Fetch}
\newcommand{\robNmExRule}{OOO-NmEx}
\newcommand{\robLdAddrRule}{OOO-LdAddr}
\newcommand{\robLdPredRule}{OOO-LdPred}
\newcommand{\robLdReqRule}{OOO-LdReq}
\newcommand{\robLdBypassRule}{OOO-LdBypass}
\newcommand{\robStExRule}{OOO-StEx}
\newcommand{\robNmLdRecRetRule}{OOO-NmLdRecCom}
\newcommand{\robStRetRule}{OOO-StCom}
\newcommand{\robComRetRule}{OOO-ComCom}
\newcommand{\robDeqSbRule}{OOO-StReq}

\newcommand{\ccmLdRule}{\ccmModelName-Ld}
\newcommand{\ccmStRule}{\ccmModelName-St}

\newcommand{\flowReorderRule}{\flowModelName-Reorder}
\newcommand{\flowFlowRule}{\flowModelName-Flow}
\newcommand{\flowBypassRule}{\flowModelName-Bypass}

\newcommand{\robSSBLdReqRule}{\robSSBModelName-LdReq}
\newcommand{\robSSBStRetRule}{\robSSBModelName-StCom}
\newcommand{\robSSBComRetRule}{\robSSBModelName-ComReq}

\newcommand{\reduceRuleSpace}{\vspace{-3pt}}
\newcommand{\reduceRuleEndSpace}{\vspace{-4pt}}

\setlist{noitemsep, leftmargin=*,topsep=2pt}

\newcommand{\oooInstFetchRule}{Instruction-Fetch}
\newcommand{\oooLoadExRule}{Load-Exectution}
\newcommand{\oooLoadRemote}{Load-Remote-Read}
\newcommand{\oooInstCommitRule}{Instruction-Commit}
\newcommand{\oooStoreIssueRule}{Store-Issue}
\newcommand{\oooMemRespRule}{Memory-Response}
\newcommand{\ccmAcceptReqRule}{Accept-Request}
\newcommand{\ccmExReqRule}{Execute-Request}


\title{An Operational Framework for Specifying Memory Models using Instantaneous Instruction Execution}
%
\authorinfo{Sizhuo Zhang \and Muralidaran Vijayaraghavan \and Arvind}{MIT CSAIL}{\{szzhang, vmurali, arvind\}@csail.mit.edu}

\maketitle 


\begin{abstract}
There has been great progress recently in formally specifying the memory model of microprocessors like ARM and POWER. 
These specifications are, however, too complicated for reasoning about program behaviors, verifying compilers \etc{}, because they involve microarchitectural details like the reorder buffer (ROB), partial and speculative execution, instruction replay on speculation failure, \etc{}
In this paper we present a new \emph{Instantaneous Instruction Execution} (\IIE) framework which allows us to specify weak memory models in the same style as SC and TSO.
Each instruction in \IIE{} is executed instantaneously and in-order such that the state of the processor is always correct. 
The effect of instruction reordering is captured by the way data is moved between the processors and the memory non-deterministically, using three conceptual devices: \emph{invalidation buffers}, \emph{timestamps} and \emph{dynamic store buffers}.
We prove that \IIE{} models capture the behaviors of modern microarchitectures and cache-coherent memory systems accurately, thus eliminating the need to think about microarchitectural details.

\end{abstract}




\category{C.0}{General}{Modeling of computer architecture} 


\keywords
Weak memory models, Operational semantics

\section{Introduction}

Computer architects make microarchitectural optimizations in processors which ensure that single-threaded programs can be run unmodified, but often create new and unexpected behaviors for multi-threaded programs. 
The effect of these optimizations manifests itself through load and store instructions because these are the only instructions through which threads can communicate with each other.
Memory models abstract hardware in a way that is useful for programmers to understand the behaviors of their programs.

There are ongoing efforts to specify memory models for multithreaded programming in C, C++ \cite{c++n4527} and other languages. 
These efforts are influenced by the type of memory models that can be supported efficiently on existing architectures like x86, POWER and ARM. 
While the memory model for x86 \cite{owens2009better,sewell2010x86,Sarkar:2009:SXM:1594834.1480929} is captured succinctly by the 
Total Store Order (TSO) model, the models for POWER \cite{sarkar2011understanding} and ARM \cite{flur2016modelling} are considerably more complex. 
The formal specifications of the POWER and ARM models have required exposing microarchitectural details like speculative execution, instruction reordering and the state of partially executed instructions, which, in the past, have always been hidden from the user. 
In addition, details of the memory system like write-though vs write-back caches, shared vs not-shared memory buffers, \etc{} were also needed for a precise specification of these two models.


Even though empirical evidence is weak, many architects believe that weak memory models, such as the ones for ARM, POWER, Alpha and RMO, offer some performance advantage or simpler implementation over TSO and other stronger memory models. 
We think that architects are unlikely to give up on weak memory models because of the flexibility they provide for high performance implementations.
It is, therefore, important to develop a framework for defining weak memory models, which, like SC and TSO operational models, does not involve microarchitecture and memory system details.   
This paper offers such a framework based on \emph{Instantaneous Instruction Execution} (\IIE). 

In the \IIE{} framework, instructions are executed in order and atomically, and consequently, the processor always has the up-to-date state.
The model descriptions use a multi-ported monolithic memory which executes loads and stores instantaneously. 
The data movement between the processors and the memory takes place asynchronously in the background.  
For specifying weak memory models, we combine \IIE{} with three new conceptual devices: \emph{invalidation buffers} to capture instruction reordering, \emph{timestamps} to enforce data dependencies, and \emph{dynamic store buffers} to model shared store buffers and write-through caches in a topology independent way. 
We present several different weak memory models -- WMM and \wmmDep{} which are similar to the Alpha and RMO models; and \wmmSSB{} which is similar to the ARM and POWER models. 

To convince the reader that we have not ruled out any fundamental and important microarchitectural optimizations, we give an abstract description of a speculative microarchitecture (\robModelName) with a coherent pipelined memory system (\ccmModelName).
The structure of \robModelName{} 
is common to all high-performance processor implementations, regardless of the memory model they support; 
implementations of stronger memory models based on \robModelName{} use extra hardware checks to prevent or kill some specific memory behaviors.  
We prove that our weakest memory model, WMM, allows all sorts of microarchitecture optimizations, that is, \robccm{} $\subseteq$ WMM. 

One optimization that has been discussed in literature but has not been implemented in any commercial microprocessor yet is \emph{load-value speculation} \cite{lipasti1996value,ghandour2010potential,perais2014eole,perais2014practical}.
It allows us to predict a value for a load; the load is killed later if the predicted value does not match the load result from the memory system.
Even if load-value speculation is included, our result \robccm{} $\subseteq$ WMM holds.
Surprisingly, if value speculation is permitted in the implementation then we can also prove that WMM $\subseteq$ \robccm, that is, the WMM and \robccm{} become equivalent.   
We show via a common programming example that an extra fence needs to be inserted in WMM to enforce data-dependencies. 
This is an unnecessary cost if we know for sure that our implementation would not use value speculation.
\wmmDep{} is a slight modification of WMM to enforce ordering of data-dependent loads using timestamps. 
We also prove that \robDepModelName{} (the \robModelName{} implementation without load-value speculation) is equivalent to \wmmDep.

ARM and POWER microarchitectures use shared store-buffers and write-though caches, and unfortunately, such memory systems introduce behaviors not seen in other weak memory models. 
The ISA for these machines include ``weaker'' and ``stronger'' fences with slightly different functionality because weaker fences have smaller performance penalty than the stronger ones. 
This requires memory fence instructions to enter store buffers, muddying the clean separation between the cache-coherent memory systems and processors.   
We introduce \flowModelName, an abstract model for hierarchy of shared store buffers or write-through caches, which is adapted from the storage subsystem in the \emph{Flowing Model} of \cite{flur2016modelling}.
We define \wmmSSB, an extension of WMM, specifically to deal with such \emph{multi-copy non-atomic store} systems and show that
\robSSBflow{} $\subseteq$ \wmmSSB{}, in which \robSSBModelName{} is the processor implementation adapted from \robModelName{} to be compatible with \flowModelName.

In summary, this paper makes the following contributions:
\begin{enumerate}
\item \IIE, a new framework for describing memory models with three new conceptual devices: invalidation buffers, timestamps, and dynamic store buffers;
\item WMM and \wmmDep{} memory models which are like the RMO and Alpha models;
\item \wmmSSB{} model to embody ARM and POWER like multi-copy non-atomic stores;
\item \robModelName, an abstract description of the microarchitecture underlying all modern high performance microprocessors;
\item A proof that \robccm{} = WMM;
\item A proof that \robDepccm{} = \wmmDep; and
\item A proof that \robSSBflow{} $\subseteq$ \wmmSSB.
\end{enumerate}

\noindent\textbf{Paper organization:}
Section \ref{sec: related work} presents the related work.
Section \ref{sec: implementation} defines \robccm{}, an implementation scheme of multiprocessors.
We introduce the \IIE{} framework in Section \ref{sec: I2E}.
We use \IIE{} and invalidation buffers to define WMM in Section \ref{sec: WMM}.
Section \ref{sec: data dep} defines \wmmDep{} using timestamps to capture data dependency.
Section \ref{sec: non atomic mem} defines \wmmSSB{} using dynamic store buffers to model multi-copy non-atomic memory systems.
Section \ref{sec: conclude} offers the conclusion.

\section{Related Work} \label{sec: related work}

SC \cite{lamport1979make} is the most intuitive memory model, but naive implementations of SC suffer from poor performance.
Gharachorloo \etal{} proposed load speculation and store prefetch to enhance the performance of SC \cite{gharachorloo1991two}.
Over the years, researchers have proposed more aggressive techniques to preserve SC \cite{ranganathan1997using,guiady1999sc+,gniady2002speculative,ceze2007bulksc,wenisch2007mechanisms,blundell2009invisifence,singh2012end,lin2012efficient,gope2014atomic}. 
Perhaps because of their hardware complexity, the adoption of these techniques in commercial microprocessor has been limited. 
Instead the manufactures and researchers have chosen to present weaker memory model interfaces, \eg{} TSO \cite{sparc1992sparcv8}, PSO \cite{weaver1994sparc}, RMO \cite{weaver1994sparc}, x86 \cite{owens2009better,sewell2010x86,Sarkar:2009:SXM:1594834.1480929}, Processor Consistency \cite{goodman1991cache}, Weak Consistency \cite{dubois1986memory}, RC \cite{gharachorloo1990memory}, CRF \cite{shen1999commit}, POWER \cite{power2013version} and ARM \cite{armv7ar}.
The tutorials by Adve \etal{} \cite{adve1996shared} and by Maranget \etal{} \cite{maranget2012tutorial} provide relationships among some of these models.

The lack of clarity in the definitions of POWER and ARM memory models in their respective company documents has led some researchers to empirically determine allowed/disallowed behaviors \cite{sarkar2011understanding,mador2012axiomatic,alglave2014herding,flur2016modelling}.
Based on such observations, in the last several years, both \emph{axiomatic} models and \emph{operational} models have been developed which are compatible with each other \cite{alglave2009semantics,Alglave2011,alglave2012formal,mador2012axiomatic,alglave2014herding,sarkar2011understanding,sarkar2012synchronising,alglave2013software,flur2016modelling}.
However, these models are quite complicated; for example, the POWER axiomatic model has 10 relations, 4 types of events per instruction, and 13 complex axioms \cite{mador2012axiomatic}, some of which have been added over time to explain specific behaviors \cite{alglave2009semantics,Alglave2011,alglave2012fences,mador2012axiomatic}. 
The abstract machines used to describe POWER and ARM operationally are also quite complicated, because they require the user to think in terms of partially executed instructions~\cite{sarkar2011understanding,sarkar2012synchronising}.
In particular, the processor sub-model incorporates ROB operations, speculations, instruction replay on speculation failures, \etc{}, explicitly, which are needed to explain the enforcement of specific dependency (\ie{} data dependency).
We present an \IIE{} model \wmmDep{} in Section \ref{sec: data dep} that captures data dependency and sidesteps all these complications.
Another source of complexity is the multi-copy non-atomicity of stores, which we discuss in Section \ref{sec: non atomic mem} with our solution \wmmSSB.

Adve \etal{} defined Data-Race-Free-0 (DRF0), a class of programs where shared variables are protected by locks, and proposed that DRF0 programs should behave as SC \cite{adve1990weak}.
Marino \etal{} improves DRF0 to the DRFx model, which throws an exception when a data race is detected at runtime \cite{Marino:2010:DSE:1806596.1806636}.
However, we believe that architectural memory models must define clear behaviors for all programs, and even throwing exceptions is not satisfactory enough.

A large amount of research has also been devoted to specifying the memory models of high-level languages, \eg{} C/C++ \cite{c++n4527,boehm2008foundations,batty2011mathematizing,batty2012clarifying,Pichon-Pharabod:2016:CSR:2837614.2837616,Batty:2016:OSA:2914770.2837637,Lahav:2016:TRC:2914770.2837643,Krebbers:2014:OAS:2535838.2535878,Batty:2013:LAC:2480359.2429099,Kang:2015:FCM:2737924.2738005} and Java \cite{manson2005java,cenciarelli2007java, maessen2000improving,Bogdanas:2015:KCS:2775051.2676982,Demange:2013:PBB:2480359.2429110}.
There are also proposals not tied to any specific language \cite{Boudol:2009:RMM:1480881.1480930,Crary:2015:CRM:2676726.2676984}.
This remains an active area of research because a widely accepted memory model for high-level parallel programming is yet to emerge, while this paper focuses on the memory models of underlying hardware.

Arvind and Maessen 
have specified precise conditions for preserving store atomicity even when instructions can be reordered  \cite{arvind2006memory}. 
In contrast, the models presented in this paper do not insist on store atomicity at the program level.

There are also studies on verifying programs running under weak memory models \cite{Torlak:2010:MCA:1809028.1806635,Kuperstein:2011:PAR:1993498.1993521,Atig:2010:VPW:1706299.1706303}.
Simple memory model definitions like \IIE{} models will definitely facilitate this research area.

\section{Implementation of Modern Multiprocessors} \label{sec: implementation}

Modern multiprocessor systems consist of out-of-order processors and highly pipelined coherent cache hierarchies.
In addition to pipelining and out-of-order execution, the processor may perform \emph{branch prediction}, \ie{} predict the PC of the next instruction during instruction fetch in order to fetch and execute the next instruction, \emph{memory dependency speculation}, \ie{} issue a load to memory even when there is an older store with unresolved address, and even \emph{load-value speculation}, \ie{} predict the result of a load before the load is executed.
The memory systems also employ pipelining and out-of-order execution for performance.
For example, the memory system may not process requests in the FIFO manner (consider a \emph{cache miss} followed by a \emph{cache hit}).
These optimizations are never visible to a single-threaded program but can be exposed by multithreaded programs.
In this section, we present ``physical'' models that describe the operations (\eg{} the ones mentioned above) inside high-performance processors and cache hierarchies.
These physical models are similar to those in \cite{sarkar2011understanding,flur2016modelling}, but here they only serve as a reference to capture the behaviors of modern multiprocessors precisely; we use them to verify the \IIE{} memory models proposed later.
It should be noted that the physical models are presented in an abstract manner, \eg{}, the inner structure of the branch predictor is abstracted by a function which may return any value.
The model also abstracts away resource management issues, such as register renaming, finite-size buffers and associated tags by assuming unbounded resources.

We associate a globally unique tag with each store instruction so that we can identify the store that each load reads from.
The tag is also saved in memory when the store writes the memory.
Such tags do not exist in real implementations but are needed in our model for reasons that will become clear in Section \ref{sec: rob rule}.

While the processor remains similar for all implementations, it is difficult to offer a common model for two dominant cache hierarchies.
Machines, such as Intel x86, have used write-back cache-coherent memory systems.
In contrast, ARM and POWER machines employ shared store-buffers and write-through caches in their memory systems.
We will first discuss \ccmModelName{}, model of a write-back cache-coherent memory system, and postpone the discussion of \flowModelName{}, the write-through cache system, until Section \ref{sec: non atomic mem}.

\subsection{\ccmModelName{}: the Semantics of Write-Back Cache Hierarchies} \label{sec: ccm spec}

Figure \ref{fig: pow+ccm} shows how out-of-order processors (\robModelName{}) and a write-back cache hierarchy (\ccmModelName{}) are connected together. 
A processor $i$ can send load and store requests to \ccmModelName{} by calling the following {methods} of port $i$ of \ccmModelName:
\begin{itemize}
	\item $\reqLdFunc(t^L, a)$: a load request to address $a$ with tag $t^L$. 
	\item $\reqStFunc(a, v, t^S)$: a store request that writes data $v$ to address $a$.
	$t^S$ is the globally unique tag for the store instruction.
\end{itemize}
Note that the processor also attaches a tag $t^L$ to each load request in order to associate the future load response with the requesting load instruction.
The memory sends responses back to a processor by calling the following methods of the processor:
\begin{itemize}
	\item $\respLdFunc(t^L, res, t^S)$: $res$ is the result for the load with tag $t^L$, and $t^S$ is the tag for the store that supplied $res$ to the memory.
	\item $\respStFunc(a)$: $a$ is the store address.
\end{itemize}
Store response is needed to inform the processor that a store has been completed. 
No ordering between the processing of requests inside the memory should be assumed by the processor.

\begin{figure}[!htb]
	\centering
	\includegraphics[width=0.8\columnwidth]{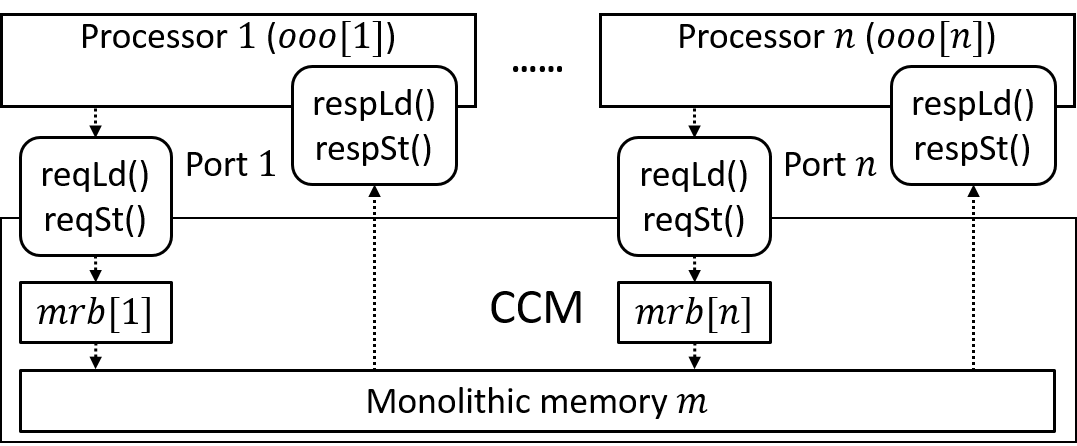}
	\nocaptionrule\caption{Multiprocessor system} \label{fig: pow+ccm}
\end{figure}

\ccmModelName{} consists of $n$ \emph{memory request buffers} $mrb[1\ldots n]$, one for each processor, and a monolithic memory $m$.
A monolithic memory location $m[a]$ contains $\langle v, t^S\rangle$, in which $v$ is the value written by a store with tag $t^S$.
The $\reqLdFunc$ and $\reqStFunc$ methods simply insert the incoming requests from processor $i$ into $mrb[i]$.
\ccmModelName{} processes requests by picking any request from any $mrb$.
If the request is a load $\langle \LdInst, t^L, a \rangle$ from $mrb[i]$, then \ccmModelName{} calls method $\respLdFunc(t^L, v, t^S)$ of processor $i$, where $\langle v, t^S\rangle = m[a]$.
If the request is a store $\langle \StInst, a, v, t^S\rangle$, then we update $m[a]$ to be $\langle v, t^S\rangle$, and call method $\respStFunc(a)$ of processor $i$.
The behavior of \ccmModelName{} is shown in Figure \ref{fig: ccm rule}. 

\begin{figure}[!htb]
	\centering
	\begin{boxedminipage}{\columnwidth}
		\small
		\textbf{\ccmLdRule{} rule} (load request processing). \reduceRuleSpace
		\begin{displaymath}
		\frac{
			\langle \LdInst, t^L, a \rangle = mrb[i].\getRandFunc();\ \langle v, t^S \rangle = m[a];
		}{
			mrb[i].\removeFunc(\langle \LdInst, t^L, a \rangle);\ ooo[i].\respLdFunc(t^L, v, t^S);
		}
		\end{displaymath} \reduceRuleEndSpace
		
		\textbf{\ccmStRule{} rule} (store request processing). \reduceRuleSpace
		\begin{displaymath}
		\frac{
			\langle \StInst, a, v, t^S \rangle = mrb[i].\getRandFunc();
		}{
			mrb[i].\removeFunc(\langle \StInst, a, v, t^S \rangle);\ m[a] \assignVal \langle v, t^S\rangle;\ ooo[i].\respStFunc(a);
		}
		\end{displaymath}
	\end{boxedminipage}
	\nocaptionrule\caption{\ccmModelName{} operational semantics} \label{fig: ccm rule}
\end{figure}

We describe the behavior of a system as a set of state-transition \emph{rules}, written as
\[
\frac{\mathit{predicates\ on\ the\ current\ state}}{\mathit{the\ action\ on\ the\ current\ state}}
\]
The predicates are expressed either by pattern matching or using a $\whenFunc(\mathit{expression})$ clause.
$mrb[i].\getRandFunc()$ returns any entry in $mrb[i]$, and $mrb[i].\removeFunc(en)$ removes entry $en$ from $mrb[i]$.

To understand how such a simple structure can abstract the write-back cache-coherent hierarchy, we refer to the cache coherence proof by Vijayaraghavan \etal{} \cite{Vijayaraghavan2015}.
It shows that a request can complete only by reading or writing an L1 cache line when it has sufficient permissions, and that under such circumstances a write-back cache-coherent hierarchy is exactly equivalent to the monolithic memory abstraction.
However, the order of responses may be different from the order of requests due to the out-of-order processing inside the hierarchy.
Such reordering is captured by $mrb$.

\subsection{\robModelName{}: the Model of Out-of-Order Processors} \label{sec: pow}

We will first give an informal description of the behavior of a speculative out-of-order processor \robModelName, shown in Figure \ref{fig: pow model}; the actual rules are presented later.

The processor fetches an instruction from the address given by the PC register, and updates PC based on the prediction by a branch predictor.
The fetched instruction is decoded and enqueued into the reorder buffer (ROB).
ROB contains all the in-flight instructions in the fetched order but executes them out of order.
An instruction can be executed when all of its source operands have been computed, and the result of its execution is stored in its ROB entry.
The computed source operands come from either an older ROB entry or the register file.
The ROB then commits the instructions in the fetched order. In-order commitment is required to implement \emph{precise interrupts} and \emph{exceptions}.
After an instruction is committed, it is removed from the ROB and the register file is updated with the result of the instruction's execution.

When a branch instruction is executed, if the branch target is not equal to the address of the next instruction that was fetched, then all instructions in ROB after the branch are ``flushed'' (\ie{} discarded) and the PC is set to the correct branch target, allowing the correct set of instructions to be fetched.

A store instruction is executed by computing the store address and data, and is enqueued into the store buffer at commit.
In the background, the store buffer can send the oldest store for an address into the memory, and delete that store when the response comes back from the memory.

In contrast, the execution of a load instruction splits into two phases.
The first phase is to compute the load address.
In the second phase, a load will search older entries in the ROB and the store buffer for the latest store to the same address.
If such a store is found, that store's value (and tag) is read -- this is called ``data forwarding'' or ``data bypassing''.
Otherwise, a load request is sent to the memory with a unique tag (we use the index of the ROB entry).
Eventually, the memory system can send a response back to the processor with a load result; the ROB entry for the load (identified with the tag) is updated with the result.

A load can be issued to memory at any time as long as its address is available, even when there are older unresolved branches or stores with uncomputed addresses.
If an older store is executed later and writes to the same address, then any such loads that were executed earlier have violated \emph{memory dependency} and should be flushed.
The details will be discussed later.
Note that loads which have been issued to the memory can be flushed from ROB for various reasons, and the responses for the flushed loads are discarded.

The processor may also employ a load-value predictor, which predicts the result of any load that does not have a value.
The predicted result can be used in the execution of other instructions.
When the load gets its value from data forwarding or memory and the value is not equal to the predicted one,  all instructions younger than the load are flushed.

There are two fences: $\ComInst$ and $\RecInst$.
The $\ComInst$ fence stalls at the commit slot until the store buffer is empty.
The $\RecInst$ fence prevents a younger load from being issued to memory, and also stops the data forwarding across the fence.

\begin{figure}[!htb]
	\centering
	\includegraphics[width=0.37\textwidth]{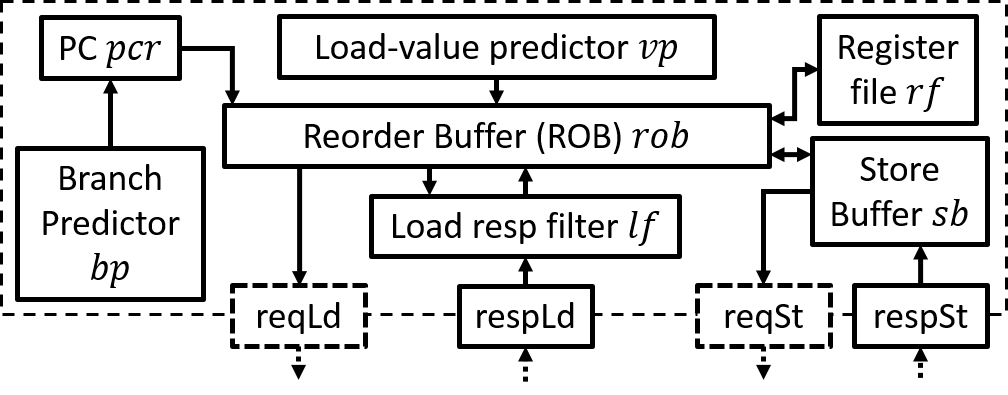}
	\nocaptionrule\caption{\robModelName{} implementation structure} \label{fig: pow model}
\end{figure}

In the following we give a precise description of how \robModelName{} operates.
(We will directly use the variable names in Figure \ref{fig: pow model}, \eg{}, $pcr$ stands for the PC register).
We never reuse ROB entries, and follow the convention that \emph{younger} entries will have \emph{larger} indices.
We refer to the oldest entry in ROB as the \emph{commit slot}.

\subsubsection{Component Functionality}
Since the implementation has to deal with partially executed instructions, we need to keep information about in-flight instructions in ROB (all $srcs$ fields represent source register names):
\begin{itemize}
	\item $\langle \NmInst, op, srcs, dst, val \rangle$: A non-memory instruction (\eg{} ALU and branch instructions).
	$op$ is the type of operation. 
	$val$ represents the computed value for the destination register $dst$ and is initially $\epsilon$. 
	These instructions include branch instructions.
	\item $\langle \LdInst, srcs, dst, a, v, t \rangle$: A load instruction to address $a$.
	$v$ is the load result for destination register $dst$, and $t$ is the tag of the store that provides value $v$. 
	All of $a$, $v$ and $t$ are initially $\epsilon$.
	\item $\langle \StInst, srcs, a, v, t \rangle$: A store instruction that writes data $v$ to address $a$.
	$t$ is the unique tag for this store assigned at decode time.
	Both $a$ and $v$ are initially $\epsilon$.
	\item $\langle \ComInst \rangle$ and $\langle \RecInst \rangle$ fences.
\end{itemize}
We use the $\fetchFunc(pc)$ function to fetch an instruction from memory address $pc$ and decode it into the above form.

$pcr$ is updated speculatively when a new instruction $ins$ is fetched into ROB using the $\predictFunc(ins)$ method of $bp$.
We assume that $bp$ always predicts correctly for non-$\NmInst$ instructions.
$\rf$ is updated conservatively when an instruction $ins$ is committed from ROB using the $\updateFunc(ins)$ method.
Each $sb$ entry also contains an $iss$ bit, which indicates whether the store has been issued to the memory or not.
The following methods are defined on $sb$:
\begin{itemize}
	\item $\enqFunc(a, v, t)$: enqueues the $\langle \mathrm{address, value, tag} \rangle$ tuple $\langle a, v, t \rangle$ into $sb$, and initializes the $iss$ bit of the new entry as $\False$.
	\item $\emptyFunc()$: returns $\True$ when $sb$ is empty.
	\item $\getAnyAddr()$: returns any store address present in $sb$; or returns $\epsilon$ if $sb$ is empty.
	\item $\getOldestFunc(a)$: return the $\langle$store data, tag, $iss$ bit$\rangle$ of the oldest store for address $a$ in $sb$.
	\item $\issueFunc(a)$: sets the $iss$ bit of the oldest store for address $a$ to $\True$.
	\item $\removeOldestFunc(a)$: deletes the oldest store to address $a$ from $sb$.
\end{itemize}

An ROB entry is defined as $\langle pc, npc, ins, ex \rangle$, where $pc$ is the PC of the instruction in the entry, $npc$ is the (predicted) PC of the next instruction, $ins$ is the instruction in this entry, and $ex$ is the state of the instruction.
$ex$ field has one of the following values: $\Idle$, $\Exe$, $\ReEx$, and $\Done$.
An instruction is $\Idle$ before it starts execution, and will become $\Done$ after execution finishes.
Both $\Exe$ and $\ReEx$ are only used for $\LdInst$ instructions to indicate that the load request is being processed in \ccmModelName{}. 
$\ReEx$ additionally implies that the load needs to be re-executed because the result of the current load request in memory is going to be wrong.
We initialize the $ex$ field of an instruction $ins$ using function $\initExFunc(ins)$, which returns $\Done$ for fence instructions and returns $\Idle$ otherwise.

$\lf$ is for filtering out load responses from \ccmModelName{} for the killed instructions. 
It is a bit vector of the same length as $rob$ (so it is also infinitely long in our description). 
$\lf[idx]$ is $\True$ if and only if a $\LdInst$ instruction at $idx$ has been flushed from $rob$ while its load request is still being processed in the memory.

Load-value speculation is modeled by the $\predictFunc(en)$ method of $vp$, which can predict the load result for ROB entry $en$.

All methods defined for $rob$ are listed in Table \ref{tab: part rob method}.
Besides, we use $rob[idx]$ to refer to the ROB entry at index $idx$. 
We also use $rob[idx].pc$, $rob[idx].npc$, $rob[idx].ins$ and $rob[idx].ex$ to refer to the $pc$, $npc$, $ins$ and $ex$ fields of $rob[idx]$, respectively.

\begin{table}[!htb]
	\centering
	\small
	\begin{tabular}{|p{0.95\columnwidth}|}
	\hline
		$\enqFunc(en)$: enqueues a new entry $en$ into $rob$. \\ \hline
		$\getReadyFunc()$: finds an entry for which all source register values are ready, and returns the $\langle \mathrm{index, entry} \rangle$ pair of it. \\ \hline
		$\getLdFunc()$: finds any entry containing a $\LdInst$ instruction, and returns the $\langle \mathrm{index, entry} \rangle$ pair of it. \\ \hline
		$\getCommitFunc()$: returns the commit slot of $rob$. \\ \hline
		$\deqFunc()$: deletes the entry, which is the commit slot, from $rob$. \\ \hline
		$\computNmFunc(idx)$: computes the result of the $\NmInst$ instruction at $idx$ using values in $\rf$ and $rob$, and returns $\langle$next PC, computed $val$ field$\rangle$. \\ \hline
		$\computeAddrFunc(idx)$: computes and returns the address of the memory instruction at $idx$ using values in $\rf$ and $rob$. \\ \hline
		$\computeStDataFunc(idx)$: computes and returns the store data of the $\StInst$ instruction at $idx$ using values in $\rf$ and $rob$. \\ \hline
		$\flushFunc(idx, pc)$: it deletes all entries in $rob$ with indices $\geq idx$ and updates $pcr$ to $pc$.
		For every deleted $\LdInst$ whose $ex$ field is $\Exe$ or $\ReEx$ (\ie{} the load is in memory), it sets the corresponding bit in $\lf$ to $\True$. \\ \hline
		$\findRecStFunc(idx, a)$: it looks for the value for $\LdInst\ a$ at $idx$ by searching older entries in the $rob$ and then in $sb$, and returns a $\langle\mathrm{value, tag}\rangle$ pair if it finds an executed store to address $a$.
		If the search finds nothing, then a $\top$ is returned so that the load can be issued to memory. 
		The search is also terminated if a $\RecInst$ fence is encountered and in that case $\epsilon$ is returned to indicate that the load should be stalled. \\ \hline
		$\findLdKilledByStFunc(idx, a)$: this method is called when a store at $idx$ resolves its address to $a$. 
		It identifies $\LdInst\ a$ instructions in $rob$ affected by this store by searching for $\LdInst\ a$ instructions from $idx+1$ to the youngest entry. 
		The search stops if another $\StInst\ a$ is encountered. 
		Since there can be several affected loads, it returns a list of their indices. 
		If a load has not started execution yet (\ie{} $ex$ field is $\Idle$), it will not be affected by $\StInst\ a$ and thus will not be returned. \\ \hline
		$\findLdKilledByLdFunc(idx, a, t)$: this method is called when a $\LdInst\ a$ at $idx$ reads from a store with tag $t$ in memory.
		It identifies $\LdInst\ a$ instructions in $rob$ which are younger than the load at $idx$ but read values staler than the value of store $t$.
		The method searches from $idx+1$ to the youngest entry for the first executed $\LdInst\ a$ instruction (\ie{} $ex$ field is $\Done$), which reads from a store with tag $t' \neq t$, and returns the index of that instruction in $rob$.
		The method returns $\top$ if no such load is found or a $\StInst\ a$ is encountered first. \\ \hline
\end{tabular}
\nocaptionrule\caption{Methods for $rob$} \label{tab: part rob method}
\end{table}

\subsubsection{Rules to Describe \robModelName{} Behavior} \label{sec: rob rule}
Figure \ref{fig: rob rule} shows the rules of \robModelName{}, where $ccm$ represents the \ccmModelName{} port connected to the processor, and Figure \ref{fig: rob interface} shows the interface methods of \robModelName{} to process the responses from memory.

\begin{figure}[!htb]
	\centering
	\begin{boxedminipage}{\columnwidth}
		\small
		\textbf{\robFetchRule{} rule} (instruction fetch). \reduceRuleSpace
		\begin{displaymath}
		\frac{
			ins = \fetchFunc(pcr);\ npc = bp.\predictFunc(ins);
		}{
			rob.\enqFunc(\langle pcr, npc, ins, \initExFunc(ins) \rangle);\ pcr \assignVal npc;
		}
		\end{displaymath} \reduceRuleEndSpace
		
		\textbf{\robNmExRule{} rule} ($\NmInst$, non-memory instruction, execution). \reduceRuleSpace
		\begin{displaymath}
		\frac{
			\begin{array}{c}
			\langle idx, \langle pc, npc, \langle \NmInst, op, srcs, dst, \epsilon \rangle, \Idle \rangle \rangle = rob.\getReadyFunc(); \\
			\langle nextpc, val \rangle = rob.\computNmFunc(idx);
			\end{array}
		}{
			\begin{array}{c}
			rob[idx] \assignVal \langle pc, nextpc, \langle \NmInst, op, srcs, dst, val \rangle, \Done \rangle; \\
			\ifFunc\ nextpc \neq npc\ \thenFunc\ rob.\flushFunc(idx+1, nextpc); \\ 
			\end{array}
		}
		\end{displaymath} \reduceRuleEndSpace
		
		\textbf{\robLdAddrRule{} rule} ($\LdInst$ address calculation). \reduceRuleSpace
		\begin{displaymath}
		\frac{
			\begin{array}{c}
			\langle idx, \langle pc, npc, \langle \LdInst, srcs, dst, \epsilon, v, \epsilon \rangle, \Idle \rangle \rangle = rob.\getReadyFunc(); \\
			a = rob.\computeAddrFunc(idx);
			\end{array}
		}{
			rob[idx].ins \assignVal \langle \LdInst, srcs, dst, a, v, \epsilon \rangle;
		}
		\end{displaymath} \reduceRuleEndSpace
		
		\textbf{\robLdPredRule{} rule} ($\LdInst$ result value prediction). \reduceRuleSpace
		\begin{displaymath}
		\frac{
			\begin{array}{c}
			\langle idx, \langle pc, npc, \langle \LdInst, srcs, dst, a, \epsilon, \epsilon \rangle, ex \rangle \rangle = rob.\getLdFunc(); \\ 
			v = vp.\predictFunc(\langle pc, npc, \langle \LdInst, srcs, dst, a, \epsilon, \epsilon \rangle, ex \rangle); \\
			\end{array}
		}{
			rob[idx].ins \assignVal \langle \LdInst, srcs, dst, a, v, \epsilon \rangle;
		}
		\end{displaymath} \reduceRuleEndSpace
		
		\textbf{\robLdBypassRule{} rule} ($\LdInst$ execution by data forwarding). \reduceRuleSpace
		\begin{displaymath}
		\frac{
			\begin{array}{c}
			\langle idx, \langle pc, npc, \langle \LdInst, srcs, dst, a, v, \epsilon \rangle, \Idle \rangle \rangle = rob.\getReadyFunc(); \\
			\langle res, t \rangle = rob.\findRecStFunc(idx, a); \\ 
			\whenFunc(a \neq \epsilon\ \wedge\ res \neq \epsilon\ \wedge\ res \neq \top); \\
			\end{array}
		}{
			\begin{array}{c}
			\qquad rob[idx] \assignVal \langle pc, npc, \langle \LdInst, srcs, dst, a, res, t \rangle, \Done \rangle; \\
			\qquad \ifFunc\ v \neq \epsilon\ \wedge\ v \neq res\ \thenFunc\ rob.\flushFunc(idx+1, npc); \\
			\end{array}
		}
		\end{displaymath} \reduceRuleEndSpace
		
		\textbf{\robLdReqRule{} rule} ($\LdInst$ execution by sending request to \ccmModelName). \reduceRuleSpace
		\begin{displaymath}
		\frac{
			\begin{array}{c}
			\langle idx, \langle pc, npc, \langle \LdInst, srcs, dst, a, v, \epsilon \rangle, \Idle \rangle \rangle = rob.\getReadyFunc(); \\
			res = rob.\findRecStFunc(idx, a); \\ 
			\whenFunc(a \neq \epsilon\ \wedge\ res == \top\ \wedge\ \neg \lf[idx]); \\
			\end{array}
		}{
			\begin{array}{c}
			ccm.\reqLdFunc(idx, a);\ rob[idx].ex \assignVal \Exe; \\
			\end{array}
		}
		\end{displaymath} \reduceRuleEndSpace
		
		\textbf{\robStExRule{} rule} ($\StInst$ execution). \reduceRuleSpace
		\begin{displaymath}
		\frac{
			\begin{array}{c}
			\langle idx, \langle pc, npc, \langle \StInst, srcs, \epsilon, \epsilon, t \rangle, \Idle \rangle \rangle = rob.\getReadyFunc(); \\
			a = rob.\computeAddrFunc(idx);\ v = rob.\computeStDataFunc(idx); \\
			list = rob.\findLdKilledByStFunc(idx, a); \\
			\end{array}
		}{
			\begin{array}{c}
			rob[idx] \assignVal \langle pc, npc, \langle \StInst, srcs, a, v, t \rangle, \Done \rangle; \\
			\begin{array}{l}
			\mathbf{for}\ kIdx = list.\mathsf{first}()\ \mathbf{to}\ list.\mathsf{tail}() \\
			\qquad \ifFunc\ rob[kIdx].ex == \Done\ \thenFunc \\
			\qquad \qquad rob.\flushFunc(kIdx, rob[kIdx].pc);\ \mathbf{break}; \\
			\qquad \elseFunc\ rob[kIdx].ex \assignVal \ReEx; \\
			\end{array} \\
			\end{array}
		}
		\end{displaymath} \reduceRuleEndSpace
		
		\textbf{\robNmLdRecRetRule{} rule} (commit $\NmInst/\LdInst/\RecInst$ from ROB). \reduceRuleSpace
		\begin{displaymath}
		\frac{
			\langle pc, npc, ins, \Done \rangle = rob.\getCommitFunc();\ \whenFunc(\isNmLdRecFunc(ins));
		}{
			rob.\deqFunc();\ \rf.\updateFunc(ins);
		}
		\end{displaymath} \reduceRuleEndSpace
		
		\textbf{\robStRetRule{} rule} (commit $\StInst$ from ROB). \reduceRuleSpace
		\begin{displaymath}
		\frac{
			\langle pc, npc, ins, \Done \rangle = rob.\getCommitFunc();\ \langle \StInst, srcs, a, v, t \rangle = ins;
		}{
			rob.\deqFunc();\ sb.\enqFunc(a, v, t);\ \rf.\updateFunc(ins);
		}
		\end{displaymath} \reduceRuleEndSpace
		
		\textbf{\robComRetRule{} rule} (commit $\ComInst$ from ROB). \reduceRuleSpace
		\begin{displaymath}
		\frac{
			\langle pc, npc, \langle \ComInst \rangle, \Done \rangle = rob.\getCommitFunc();\ \whenFunc(sb.\emptyFunc());
		}{
			rob.\deqFunc();\ \rf.\updateFunc(\langle \ComInst \rangle);
		}
		\end{displaymath} \reduceRuleEndSpace
		
		\textbf{\robDeqSbRule{} rule} (issue store to \ccmModelName). \reduceRuleSpace
		\begin{displaymath}
		\frac{
			a = sb.\getAnyAddr();\ \langle v, t, \False \rangle = sb.\getOldestFunc(a);\ \whenFunc(a \neq \epsilon);
		}{
			ccm.\reqStFunc(a, v, t);\ sb.\issueFunc(a);
		}
		\end{displaymath}
	\end{boxedminipage}
	\nocaptionrule\caption{\robModelName{} operational semantics} \label{fig: rob rule}
\end{figure}

\begin{figure}[!htb]
	\centering
	\begin{boxedminipage}{\columnwidth}
		\small
		$\respLdFunc(idx, res, t)$ \textbf{method:} \reduceRuleSpace
		\begin{displaymath}
		\begin{array}{l}
		\ifFunc\ \lf[idx]\ \thenFunc\ \lf[idx] \assignVal \False;\ \textcolor{blue}{/\!/\ \mathrm{wrong\Hyphen{}path\ load\ response}}  \\
		\elseFunc \\
		\quad \letFunc\ \langle pc, npc, \langle \LdInst, srcs, dst, a, v, \epsilon \rangle, ex \rangle = rob[idx]\ \inFunc \\
		\quad \quad \ifFunc\ ex == \ReEx\ \thenFunc\ rob[idx].ex \assignVal \Idle; \\
		\quad \quad \elseFunc\ \codecomment{/\!/\ \mathrm{save\ load\ result\ and\ check\ value\ misprediction}} \\
		\quad \quad \quad rob[idx] \assignVal \langle pc, npc, \langle \LdInst, srcs, dst, a, res, t \rangle, \Done \rangle; \\
		\quad \quad \quad \ifFunc\ v \neq \epsilon\ \wedge\ v \neq res\ \thenFunc\ rob.\flushFunc(idx+1, npc); \\
		\quad \quad \quad \elseFunc\ \codecomment{/\!/\ \mathrm{kill\ younger\ load\ with\ staler\ value}} \\
		\quad \quad \quad \quad kIdx = rob.\findLdKilledByLdFunc(a,idx,t); \\
		\quad \quad \quad \quad \ifFunc\ kIdx \neq \top\ \thenFunc\ rob.\flushFunc(kIdx, rob[kIdx].pc); \\
		\end{array}
		\end{displaymath} \reduceRuleEndSpace
		
		$\respStFunc(a)$ \textbf{method:} $sb.\removeOldestFunc(a);$
	\end{boxedminipage}
	\nocaptionrule\caption{\robModelName{} interface methods} \label{fig: rob interface}
\end{figure}

Rule \robLdReqRule{} sends a load request to \ccmModelName{}. 
However, $\lf$ has to be checked to avoid sending a request with a duplicated $idx$ to memory.
(Since we never reuse $rob$ indices here, this check will always pass).
When the load response arrives as shown in the $\respLdFunc(idx, res)$ method in Figure \ref{fig: rob interface}, we check $\lf$ to see if it corresponds to a $\LdInst$ instruction which has already been killed.
If so, we throw away the response and reset the $\lf$ bit.
Otherwise, we check the $ex$ field of the original requesting ROB entry.
If it is $\ReEx$, the load response is discarded and the $ex$ field is set to $\Idle$ so that the load can be re-executed later.
Otherwise, we record the load result and flush $rob$ in case of load-value misprediction.
If there is no load-value misprediction, we kill eagerly executed younger loads which get results staler than the current load response using the $\findLdKilledByLdFunc$ method.

There are two points worth noticing about loads.
First, the load address can be computed from some unverified predicted values, so two loads can be executed out-of-order even if one load uses the result of the other as load address.
Second, two loads to the same address on the same processor can return from \ccmModelName{} out-of-order as long as they reads from the same store, making the loads still appear to be executed in order.
While this mechanism assumes that the load result has the unique tag associated with the store read by the load, in a real implementation there are no unique tags for stores.
In actual implementations, the processors monitors the coherence messages during the period between these two load responses; if the cache-line read by the younger load is invalidated, then the younger load is killed.
This mechanism helps maintain the SC for a single address property at the program level while imposing minimum restrictions on the out-of-order execution in hardware.
POWER and ARM processors also employ this mechanism \cite{sarkar2011understanding,flur2016modelling}.
(Notice that the tags for stores are solely for the purpose of detecting whether two loads read from the same store).
We do not use $\findLdKilledByLdFunc$ to kill loads in \robLdBypassRule{}, because such loads must have already been killed by the store that forwards the data, as explained below in the \robStExRule{} rule.

Rule \robStExRule{} computes the store address and data of a $\StInst$ instruction, and searches for younger loads to the same address which violate memory dependency.
The \textbf{for} loop is used to process all violating loads from the oldest to the youngest.
If the load has not finished execution, we mark its $ex$ field for re-execution. 
Otherwise, the load may have propagated the wrong result to younger instructions, so we kill it and flush the subsequent instructions from $rob$.

\section{Defining Memory Models Using I\textsuperscript{2}E} \label{sec: I2E}

The \IIE{} framework defines memory models using the structure in Figure \ref{fig: gen op model}.
The state of the system with $n$ processors is defined as $\langle ps, m \rangle$, where $m$ is an $n$-ported {monolithic} memory which is connected to the $n$ processors. 
$ps[i]$ ($i=1\ldots n$) represents the state of the $i^{th}$ processor. 
Each processor contains a register state $s$, which represents all architectural registers, including both the general purpose registers and special purpose registers, such as PC. 
Cloud represents additional state elements, \eg{} a \emph{store buffer},  that a specific memory model may use in its definition.

Since we want our definitions of the memory models to be independent from ISA, we introduce the concept of \emph{decoded instruction set} (DIS).
A \emph{decoded instruction} contains all the information of an instruction after it has been decoded and has read all source registers.
Our DIS has the following five instructions. 
\begin{itemize}
	\item $\langle \NmInst, dst, v \rangle$: instructions that do not access memory, such as ALU or branch instructions.
	It writes the computation result $v$ into destination register $dst$.
	\item $\langle \LdInst, a, dst \rangle$: a load that reads memory address $a$ and updates the destination register $dst$.
	\item $\langle \StInst, a, v \rangle$: a store that writes value $v$ to memory address $a$.
	\item $\langle \ComInst\rangle$ and $\langle \RecInst \rangle$: the two types of fences.
\end{itemize}
The $\ComInst$ and $\RecInst$ fences already appeared in \robModelName, and later we will define their semantics in proposed \IIE{} models.

Since instructions are executed instantaneously and in-order in \IIE{} models, the register state of each processor is by definition always up-to-date.
Therefore we can define the following two methods on each processor to manipulate the register state $s$:
\begin{itemize}
	\item $\decodeFunc()$: fetches the next instruction and returns the corresponding decoded instruction based on the current $s$.
	\item $\executeFunc(dIns, ldRes)$: updates $s$ (\eg{} by writing destination registers and changing PC) according to the current decoded instruction $dIns$.
	A $\LdInst$ requires a second argument $ldRes$ which should be the loaded value.
	For other instructions, the second argument can be set to don't care (``$\Hyphen$").
\end{itemize}

In \IIE{}, the meaning of an instruction cannot depend on a future store, so all \IIE{} models forbid stores from overtaking loads.
This automatically excludes all thin-air read behaviors, and matches the in-order commit property of \robModelName.

\begin{figure}[!htb]
	\centering
	\subfloat[General model structure\label{fig: gen op model}]{\includegraphics[width=0.535\columnwidth]{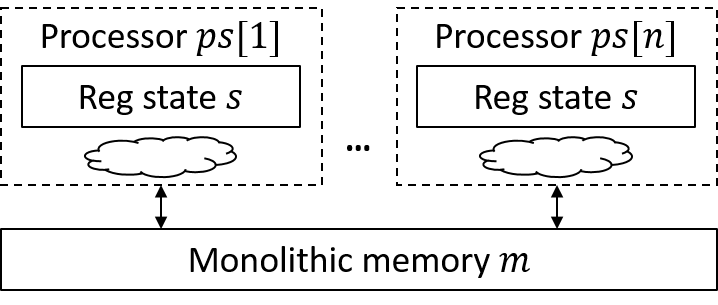}} \quad
	\subfloat[WMM structure\label{fig: wmm op model}]{\includegraphics[width=0.385\columnwidth]{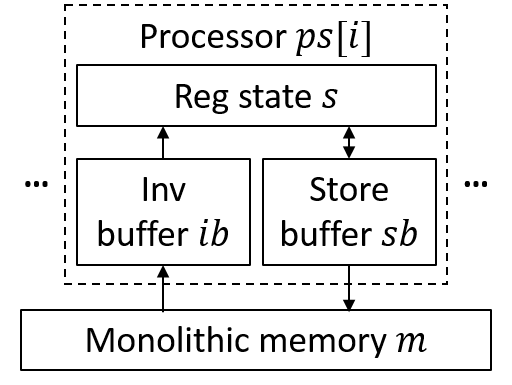}}
	\nocaptionrule\caption{Structures of \IIE{} models}
\end{figure}

\section{WMM Model} \label{sec: WMM}

Our first \IIE{} model, WMM (Figure \ref{fig: wmm op model}), adds two conceptual devices to each processor: a store buffer $sb$ and an invalidation buffer $ib$.
Despite the simplicity of these two devices, they make WMM equivalent to \robccm.
It is notable that WMM can capture the subtle effects induced by various speculations in \robModelName.

The $sb$ in WMM is almost the same as the one in \robModelName{} except that it does not need the $\issueFunc$ method here and thus the $iss$ bit is also not needed.
(The store tag is also not needed here).
We change the $\removeOldestFunc(a)$ method of $sb$ to return the $\langle$address, value$\rangle$ pair of the oldest store for address $a$ in addition to the deletion of that store.
We also define the following two new methods on $sb$:
\begin{itemize}
	\item $\existFunc(a)$: returns $\True$ if address $a$ is present in $sb$.
	\item $\getYoungestFunc(a)$: returns the youngest store data to address $a$ in $sb$.
\end{itemize}
Buffering stores in $sb$ allows loads to overtake stores, and enables reorderings of stores.
A $\ComInst$ fence will flush all stores in the local $sb$ into the monolithic memory to make them globally visible.

In case of load-load reordering, a reordered load may read a stale value, and this behavior is simulated by the $ib$ of each processor in WMM.
$ib$ is an unbounded buffer of $\langle \mathrm{address, value} \rangle$ pairs, each representing a stale memory value for an address that can be observed by the processor. 
A stale value enters $ib$ when some store buffer pushes a value to the monolithic memory.
When ordering is needed, stale values should be removed from $ib$ to prevent younger loads from reading them.
In particular, the $\RecInst$ fence will clear the local $ib$.
The following methods are defined on $ib$:
\begin{itemize}
	\item $\insertFunc(a,v)$: inserts $\langle \mathrm{address, value} \rangle$ pair $\langle a, v \rangle$ into $ib$.
	\item $\existFunc(a)$: returns $\True$ if address $a$ is present in $ib$.
	\item $\getRandAndRemoveFunc(a)$: returns any value $v$ for address $a$ in $ib$, and removes all values for $a$, which are inserted into $ib$ before $v$, from $ib$.
	\item $\clearFunc()$: removes all contents from $ib$ to make it empty.
	\item $\removeAddrFunc(a)$: removes all (stale) values for address $a$ from $ib$.
\end{itemize}

\subsection{Operational Semantics of WMM}

Figure \ref{fig: wmm rule} shows the operational semantics of WMM.
The first 7 rules are the instantaneous execution of decoded instructions, while the \wmmDeqSbRule{} rule removes the oldest store for any address (say $a$) from $sb$ and commits it to the monolithic memory. 
\wmmDeqSbRule{} also inserts the original memory value into the $ib$ of all other processors to allow $\LdInst\ a$ in these processors to effectively get reordered with older instructions. 
However, this insertion in $ib$ should not be done if the corresponding $sb$ on that processor already has a store to $a$. 
This restriction is important, because if a processor has address $a$ in its $sb$, then it can never see stale values for $a$. 
For the same reason, when a $\StInst\ a\ v$ is inserted into $sb$, we remove all values for $a$ from the $ib$ of the same processor.

\begin{figure}[!htb]
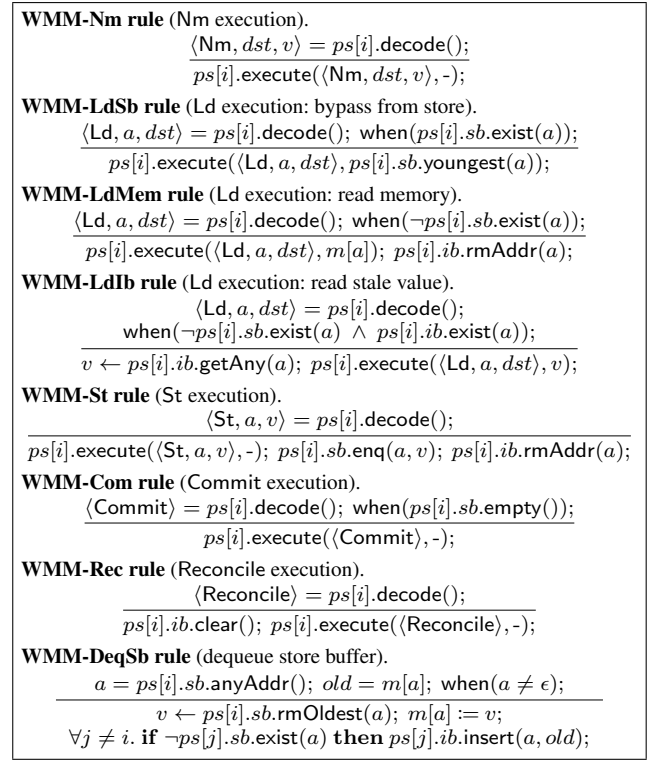

	\centering
	\begin{boxedminipage}{\columnwidth}
		\small
		\textbf{\wmmNmRule{} rule} ($\NmInst$ execution). \reduceRuleSpace
		\begin{displaymath}
		\frac{
			\langle \NmInst, dst, v \rangle = ps[i].\decodeFunc();
		}{
			ps[i].\executeFunc(\langle \NmInst, dst, v \rangle, \Hyphen);
		}
		\end{displaymath} \reduceRuleEndSpace
		
		\textbf{\wmmLdSbRule{} rule} ($\LdInst$ execution: bypass from store). \reduceRuleSpace
		\begin{displaymath}
		\frac{
			\langle \LdInst, a, dst \rangle = ps[i].\decodeFunc();\ \whenFunc(ps[i].sb.\existFunc(a));
		}{
			ps[i].\executeFunc(\langle \LdInst, a, dst \rangle, ps[i].sb.\getYoungestFunc(a));
		}
		\end{displaymath} \reduceRuleEndSpace
		
		\textbf{\wmmLdMemRule{} rule} ($\LdInst$ execution: read memory). \reduceRuleSpace
		\begin{displaymath}
		\frac{
			\langle \LdInst, a, dst \rangle = ps[i].\decodeFunc();\ \whenFunc(\neg ps[i].sb.\existFunc(a));
		}{
			ps[i].\executeFunc(\langle \LdInst, a, dst \rangle, m[a]);\ ps[i].ib.\removeAddrFunc(a);
		}
		\end{displaymath} \reduceRuleEndSpace
		
		\textbf{\wmmLdIbRule{} rule} ($\LdInst$ execution: read stale value). \reduceRuleSpace
		\begin{displaymath}
		\frac{
			\begin{array}{c}
			\langle \LdInst, a, dst \rangle = ps[i].\decodeFunc(); \\
			\whenFunc(\neg ps[i].sb.\existFunc(a)\ \wedge\ ps[i].ib.\existFunc(a)); \\
			\end{array}
		}{
			v \leftarrow ps[i].ib.\getRandAndRemoveFunc(a);\ ps[i].\executeFunc(\langle \LdInst, a, dst \rangle, v);
		}
		\end{displaymath} \reduceRuleEndSpace
		
		\textbf{\wmmStRule{} rule} ($\StInst$ execution). \reduceRuleSpace
		\begin{displaymath}
		\frac{
			\langle \StInst, a, v \rangle = ps[i].\decodeFunc();
		}{
			ps[i].\executeFunc(\langle \StInst, a, v \rangle, \Hyphen);\ ps[i].sb.\enqFunc(a, v);\ ps[i].ib.\removeAddrFunc(a);
		}
		\end{displaymath} \reduceRuleEndSpace
		
		\textbf{\wmmComRule{} rule} ($\ComInst$ execution). \reduceRuleSpace
		\begin{displaymath}
		\frac{
			\langle \ComInst \rangle = ps[i].\decodeFunc();\ \whenFunc(ps[i].sb.\emptyFunc());
		}{
			ps[i].\executeFunc(\langle \ComInst \rangle, \Hyphen);
		}
		\end{displaymath} \reduceRuleEndSpace
		
		\textbf{\wmmRecRule{} rule} ($\RecInst$ execution). \reduceRuleSpace
		\begin{displaymath}
		\frac{
			\langle \RecInst \rangle = ps[i].\decodeFunc();
		}{
			ps[i].ib.\clearFunc();\ ps[i].\executeFunc(\langle \RecInst \rangle, \Hyphen);
		}
		\end{displaymath} \reduceRuleEndSpace
		
		\textbf{\wmmDeqSbRule{} rule} (dequeue store buffer). \reduceRuleSpace
		\begin{displaymath}
		\frac{
			\begin{array}{c}
			a = ps[i].sb.\getAnyAddr();\ old = m[a];\ \whenFunc(a \neq \epsilon); \\
			\end{array}
		}{
			\begin{array}{c}
			v \leftarrow ps[i].sb.\removeOldestFunc(a);\ m[a] \assignVal v; \\
			\forall j \neq i.\ \ifFunc\ \neg ps[j].sb.\existFunc(a)\ \thenFunc\ ps[j].ib.\insertFunc(a,old); \\
			\end{array}
		}
		\end{displaymath}
	\end{boxedminipage}
	\nocaptionrule\caption{WMM operational semantics} \label{fig: wmm rule}
\end{figure}

Load execution rules in Figure \ref{fig: wmm rule} correspond to three places from where a load can get its value.
\wmmLdSbRule{} executes $\LdInst\ a$ by reading from $sb$.
If address $a$ is not found in $sb$, then the load can read from the monolithic memory (\wmmLdMemRule).
However, in order to allow the load to read a stale value (to model load reordering), \wmmLdIbRule{} gets the value from $ib$.
(Since $\getRandAndRemoveFunc$ has side-effects, we use $\leftarrow$ to bind its return value to a free variable).
The model allows {non-deterministic} choice in the selection of \wmmLdMemRule{} and \wmmLdIbRule{}.
To make this idea work, \wmmLdMemRule{} has to remove all values for $a$ from $ib$,  because these values are staler than the value in memory.
Similarly, \wmmLdIbRule{} removes all the values for $a$, which are staler than the one read, from $ib$.

\noindent\textbf{Synchronization instructions:}
Atomic read-modify-write ($\RMWInst$) instructions can also be included in WMM.
$\RMWInst$ directly operates on the monolithic memory, so the rule to execute $\RMWInst$ is simply the combination of \wmmLdMemRule, \wmmStRule{} and \wmmDeqSbRule.

\subsection{Litmus Tests for WMM} \label{sec: wmm litmus test}

WMM executes instructions instantaneously and in order, but because of store buffers ($sb$) and invalidation buffers ($ib$), a processor can see the effect of loads and stores on some other processor in a different order than the program order on that processor.
We explain the reorderings permitted and forbidden by the definition of WMM using well-known examples.

\noindent\textbf{Fences for mutual exclusion:}
Figure \ref{fig: dekker wmm} shows the kernel of Dekker's algorithm in WMM, which guarantees mutual exclusion by ensuring registers $r_1$ and $r_2$ cannot both be zero at the end.
All four fences are necessary.
Without the $\RecInst$ fence $I_3$, $I_4$ could read 0 from $ib$, as if $I_4$ overtakes $I_1$ and $I_2$.
Without the $\ComInst$ fence $I_2$, $I_1$ could stay in the $sb$, and $I_8$ gets 0 from memory.

\begin{figure}[!htb]
	\centering
	\begin{minipage}[b]{0.47\columnwidth}
		\centering
		\small
		\begin{tabular}{|l|l|}
			\hline
			{Proc. P1} & {Proc. P2} \\ 
			\hline
			$\!\! I_1: \StInst\ a\ 1 \!\!$    & $\!\! I_5: \StInst\ b\ 1 \!\!$ \\
			$\!\! I_2: \ComInst \!\!$         & $\!\! I_6: \ComInst \!\!$ \\
			$\!\! I_3: \RecInst \!\!$         & $\!\! I_7: \RecInst \!\!$ \\
			$\!\! I_4: r_1 = \LdInst\ b \!\!$ & $\!\! I_8: r_2 = \LdInst\ a \!\!$ \\ \hline
			\multicolumn{2}{|l|}{\hspace{-3pt}WMM forbids: $r_1 = 0, r_2 = 0 \!\!$} \\ \hline
		\end{tabular}
		\nocaptionrule \caption{Dekker's algorithm in WMM} \label{fig: dekker wmm}
	\end{minipage}
	\hspace{1pt}
	\begin{minipage}[b]{0.47\columnwidth}
		\centering
		\small
		\begin{tabular}{|l|l|}
			\hline
			{Proc. P1} & {Proc. P2} \\ 
			\hline
			$\!\! I_1: \StInst\ a\ 42$  & $\!\! I_4: r_1 = \LdInst\ f \!\!$ \\
			$\!\! I_2: \ComInst$        & $\!\! I_5: \RecInst \!\!$ \\
			$\!\! I_3: \StInst\ f\ 1$   & $\!\! I_6: r_2 = \LdInst\ a \!\!$ \\ 
			\hline
			\multicolumn{2}{|l|}{\hspace{-3pt}WMM forbids: $r_1=1, r_2 = 0 \!\!$} \\ 
			\hline
		\end{tabular}
		\nocaptionrule \caption{Message passing in WMM} \label{fig: msg pass wmm}
	\end{minipage}
\end{figure}

\noindent\textbf{Fences for message passing:}
Figure \ref{fig: msg pass wmm} shows a way of passing data 42 from P1 to P2 by setting a flag at address $f$.
Fences $I_2$ and $I_5$ are necessary.
Without the $\ComInst$ fence $I_2$, the two stores on P1 may get reordered in writing memory, so $I_6$ may not see the new data.
Without the $\RecInst$ fence $I_5$, $I_6$ could see the stale value 0 from $ib$.
It is as if the two loads on P2 are reordered.

\noindent\textbf{Memory dependency speculation:}
WMM is able to capture the behaviors caused by memory dependency speculation in hardware.
For example, the behavior in Figure \ref{fig: wmm mem dep pred} is possible in \robccm{} due to memory dependency speculation, \ie{} P2 predicts that the store address of $I_5$ is not $a$, and execute $I_6$ early to get value 0.
WMM allows this behavior because $I_6$ can read 0 from $ib$.

\begin{figure}[!htb]
	\begin{minipage}[b]{0.52\columnwidth}
		\small
		\begin{tabular}{|l|l|}
			\hline
			Proc. P1 & Proc. P2 \\
			\hline
			$\!\! I_1: \StInst\ a\ 1 \!\!$ & $\!\! I_4: r_1 = \LdInst\ b \!\!$ \\
			$\!\! I_2: \ComInst \!\!$      & $\!\! I_5: \StInst\ (r_1 \!+\! c \!-\! 1)\ 1 \!\!$ \\
			$\!\! I_3: \StInst\ b\ 1 \!\!$ & $\!\! I_6: r_2 = \LdInst\ a \!\!$ \\
			\hline
			\multicolumn{2}{|l|}{WMM allows: $r_1=1,r_2=0$} \\
			\hline
		\end{tabular}
		\nocaptionrule \caption{Memory dependency speculation} \label{fig: wmm mem dep pred}
	\end{minipage}
	\hspace{2pt}
	\begin{minipage}[b]{0.45\columnwidth}
		\centering
		\small
		\begin{tabular}{|l|l|}
			\hline
			Proc. P1 & Proc. P2 \\
			\hline
			$\!\! I_1: \StInst\ a\ 1 \!\!$ & $\!\! I_4: r_1 = \LdInst\ b \!\!$ \\
			$\!\! I_2: \ComInst \!\!$      & $\!\! I_5: r_2 = \LdInst\ r_1 \!\!$ \\
			$\!\! I_3: \StInst\ b\ a \!\!$ & \\
			\hline
			\multicolumn{2}{|l|}{\hspace{-4pt}WMM allows: $r_1=a,r_2=0 \!\!$} \\
			\hline
		\end{tabular}
		\nocaptionrule \caption{Load-value speculation} \label{fig: wmm val pred}
	\end{minipage}
\end{figure}

\noindent\textbf{Load-value speculation:}
WMM can capture the behaviors caused by load-value speculation in hardware.
For instance, the behavior in Figure \ref{fig: wmm val pred} is the result of such speculation in \robccm, \ie{} P2 can predict the result of $I_4$ to be $a$ and execute $I_5$ early to get value 0.
When $I_4$ returns from memory later with value $a$, the prediction on $I_4$ turns out to be correct and the result of $I_5$ can be kept.
WMM allows this behavior because $I_5$ can read 0 from $ib$.



\noindent\textbf{SC for a single address:}
WMM maintains SC for all accesses to a single address, \ie{} all loads the stores to a single address can be put into a total order, which is consistent with the program order ($\poEdge$), read-from relation ($\rfEdge$), and coherence order ($\coEdge$).
The coherence order is a total order of stores to this address; in WMM it is the order of writing the monolithic memory.
This property holds for WMM because both $sb$ and $ib$ manages values of the same address in a FIFO manner.
This property also implies the following two axioms \cite{batty2011mathematizing,maranget2012tutorial,c++n4527} ($L_1,L_2,S_1,S_2$ denote loads and stores to the same address):
\begin{description}
	\descitem{CoRR} (Read-Read Coherence): $L_1\poEdge L_2\ \wedge\ S_1\rfEdge L_1\ \wedge\ S_2\rfEdge L_2 \Longrightarrow S_1 == S_2\ \vee\ S_1\coEdge S_2$. 
	\descitem{CoWR} (Write-Read Coherence): $S_2\rfEdge L_1\ \wedge\ S_1\poEdge L_1 \Longrightarrow S_1 == S2\ \vee\ S_1\coEdge S_2$.
\end{description}
WMM satisfies these two axioms.
As for \robccm{}, The coherence order is the order of writing the monolithic memory in \ccmModelName, and these two axioms hold due to the $\findLdKilledByStFunc$ and $\findLdKilledByLdFunc$ methods used in \robModelName{} (see Appendix~\ref{sec:sup-sc} for the proof).

\noindent\textbf{SC for well-synchronized programs:}
The critical sections in well-synchronized programs are all protected by locks.
To maintain SC behaviors for such programs in WMM, we can add a $\RecInst$ after acquiring the lock and a $\ComInst$ before releasing the lock.

\noindent\textbf{Fences to restore SC:}
For any program, if we insert a $\ComInst$ followed by a $\RecInst$ before every $\LdInst$ and $\StInst$ instruction, the program behavior in WMM will be sequential consistent.

In summary, WMM can reorder stores to different addresses, and allows a load to overtake other loads (to different addresses), stores and $\ComInst$ fences.
A load cannot overtake any $\RecInst$ fence, while dependencies generally do not enforce ordering.

\subsection{Equivalence of \robccm{} and WMM}

\begin{theorem}\label{thm: rob < wmm}
	\robccm{} $\subseteq$ WMM.
\end{theorem}
\begin{proof}
	First of all, for any execution in \robccm{} which contains flushes on $rob$, if we exclude the flushes and the part of the execution, of which the effects are canceled by the flushes, we will get a new execution with the same program results.
	Similarly, we could exclude any load re-execution (\ie{} a store sets the $ex$ field of a load to $\ReEx$).
	Thus, we only need to consider executions in \robccm{} without any $rob$ flush or load re-execution.
	For any such execution $E$ in \robccm{}, we could simulate $E$ in WMM using the following way to get the same program behavior:
	\begin{itemize}
		\item When an instruction is committed from an ROB in $E$, we execute that instruction in WMM.
		\item When a store writes the monolithic memory of \ccmModelName{} ($ccm.m$) and is dequeued from a store buffer in $E$, we also dequeue that store and writes the monolithic memory ($m$) in WMM.
	\end{itemize}
	After each step of our simulation, we prove inductively that the following invariants hold:
	\begin{enumerate}
		\item The WMM states of $m$ and all store buffers are the same as the \robccm{} states of $ccm.m$ and all store buffers.
		\item All instructions committed from ROBs in \robccm{} have also been executed in WMM with the same results.
	\end{enumerate}
	The only non-trivial case is when a load $L$ to address $a$ is committed from the ROB of processor $i$ ($ooo[i].rob$) in \robccm{} and $L$ is executed correspondingly by $ps[i]$ in WMM.
	Assume that $L$ reads from store $S$ in \robccm{} (via memory, store buffer or ROB).
	We consider the status of $S$ in \robccm{} when $L$ is committed from ROB.
	If $S$ is still in the store buffer ($ooo[i].sb$) or $ccm.m$, then WMM can execute $L$ by reading from $ps[i].sb$ or $m$.
	Otherwise $S$ must have been overwritten by another store in $ccm.m$ before $L$ is committed from ROB.
	In this case, WMM will insert the value of $S$ into $ps[i].ib$ when the overwrite happens in WMM, because there cannot be any store to $a$ in $ooo[i].sb$ at the time of the overwrite.
	Now we only need to show that the value of $S$ is not removed from $ps[i].ib$ by any store, $\RecInst$ or load before $L$ is executed in WMM so that $L$ could read the value of $S$ from $ps[i].ib$ in WMM.
	We consider the time period after the overwrite and before the commit of $L$ in \robccm{}, as well as the corresponding period in WMM (\ie{} after the overwrite and before the execution of $L$).
	Since there cannot be any store to $a$ or $\RecInst$ fence committed from $ooo[i].rob$ during that period (otherwise $L$ cannot read from $S$), the value of $S$ will not be removed from $ps[i].ib$ by stores or $\RecInst$ fences during that period, and $ps[i].sb$ will not contain any store to $a$ when $L$ is executed.
	Furthermore, the \descref{CoRR} axiom of \robccm{} implies that each load $L'$ to address $a$ committed from $ooo[i].rob$ during that period must read from either $S$ or another $S$ which writes $ccm.m$ before $S$.
	Thus, the execution of $L'$ in WMM cannot remove $S$ from $ps[i].ib$, and $L$ can indeed read $S$ from $ps[i].ib$.
\end{proof}

\begin{theorem}\label{thm: rob > wmm}
	WMM $\subseteq$ \robccm.
\end{theorem}
\begin{proof}
	For any WMM execution $E$, we could construct a rule sequence $E'$ in \robccm{}, which has the same program behavior.
	The first part of $E'$ is to fetch all instructions executed in $E$ into the ROB of each processor (using \robFetchRule{} rules), then predict the value of every load to the result of that load in $E$ (using \robLdPredRule{} rules), and finally compute all $\NmInst$ instructions, load/store addresses and store data.
	The rest of $E'$ is to simulate each rule in WMM using the following way:
	\begin{itemize}
		\item When a store is dequeued from $ps[i].sb$ and written into $m$ in WMM, we also fire the \robDeqSbRule{} and \ccmStRule{} rules consecutively to write that store into $ccm.m$ and dequeues it from $ooo[i].sb$ in \robccm.
		\item When an instruction $I$ is executed by $ps[i]$ in WMM, we commit $I$ from $ooo[i].rob$ in \robccm.
		If $I$ is a $\LdInst$, we additionally schedule rules to execute this load in \robccm:
		\begin{itemize}
			\item If $I$ reads from $ps[i].sb$ in WMM, then we fire an \robLdBypassRule{} rule for $I$ right before it is committed from ROB.
			\item If $I$ reads from $m$ in WMM, then we fire an \robLdReqRule{} rule and a \ccmLdRule{} rule consecutively to execute it load right before it is committed from ROB.
			\item If $I$ reads from $ps[i].ib$, then we fire a \robLdReqRule{} rule and a \ccmLdRule{} rule consecutively to execute it right before the load value (which $I$ gets in WMM) is overwritten in $ccm.m$.
		\end{itemize}
	\end{itemize}
	Although the construction of the rule sequence is not in order, the sequence constructed after every step is always a valid rule sequence in \robccm{} for all instructions already executed by WMM.
	When we schedule rules for an instruction $I$ in \robccm{}, the constructed $E'$ does not contain any rule for instructions younger than $I$ in program order, and thus the rules for $I$ will not affect any existing rule in $E'$.
	Besides, all operations that the execution of $I$ depends on (\eg{}, executing an instruction older than $I$, or writing a store, which $I$ will read from $ccm.m$ in the scheduled rule, into $ccm.m$) are already in the constructed $E'$, so the scheduled rules for $I$ will not depend on any rule scheduled in the future construction.
	We can prove inductively that the following invariants hold after each construction step:
	\begin{enumerate}
		\item The states of $m$ and all $sb$ in WMM are the same as the states of $ccm.m$ and all $sb$ in \robccm{}.
		\item All instructions executed in WMM have also been executed (with the same results) and committed from ROBs in \robccm{}.
		\item There is no ROB flush or load re-execution in $E'$.
	\end{enumerate}
	The only non-trivial case is that a load $L$ to address $a$ reads the value of a store $S$ from $ps[i].ib$ in WMM.
	In this case, the overwrite of $S$ in $ccm.m$ must happen in the constructed $E'$, and $L$ must be in $ooo[i].rob$ at the time of overwrite.
	Since $S$ is inserted into $ps[i].ib$ and is not removed by any store or $\RecInst$ before $L$ is executed in WMM, there is no store to $a$ or $\RecInst$ fence older than $L$ in $ooo[i].rob$ or $ooo[i].sb$ right before the overwrite in \robccm{}.
	Thus the \robLdReqRule{} and \ccmLdRule{} rules are able to fire and read the value of $S$ at that time.
	The \descref{CoRR} axiom of WMM implies that any load $L'$ to $a$ which is older than $L$ must read from either $S$ or another store that writes $m$ before $S$.
	Thus $L'$ must get its result (not predicted values) before the overwrite in \robccm{}, and $L$ cannot be killed by $L'$.
	($L$ cannot be killed or re-executed by any store because all \robStExRule{} rules are fired at the beginning).
\end{proof}






\section{Modeling Data Dependency} \label{sec: data dep}

Current commercial processors do not use load-value speculation, and these processors can be modeled by \robDepccm, in which \robDepModelName{} is derived by removing the value predictor $vp$ and related operations (\eg{} the \robLdPredRule{} rule) from \robModelName.
The removal of $vp$ implies the enforcement of data-dependency ordering in hardware (\robDepccm).
For example, the behavior in Figure \ref{fig: wmm val pred} is forbidden by \robDepccm.
However, it requires inserting a $\RecInst$ fence between $I_4$ and $I_5$ to forbid this behavior in WMM.
This fence may cause performance loss because it would prevent the execution of loads that follow $I_5$ but do not depend on $I_4$.
This is an unnecessary cost if programs are running on commercial hardware captured by \robDepccm.
To avoid extra $\RecInst$ fences, we present \emph{\wmmDep{}}, an \IIE{} model equivalent to \robDepccm.
\wmmDep{} is derived from WMM by introducing \emph{timestamps} to exclude exactly the behaviors that violate data-dependency orderings.

\subsection{Enforcing Data Dependency with Timestamps} \label{sec: wmm dep intuition}

We derive our intuition for timestamps by observing how \robDepModelName{} works. 
In Figure \ref{fig: wmm val pred}, assume instruction $I_k$ ($k=1\ldots 5$) gets its result or writes memory at time $t_k$ in \robDepccm.
Then $t_5 \geq t_4$ because the result of $I_4$ is a source operand of $I_5$ (\ie{} the load address).
Since $I_4$ reads the value of $I_3$ from memory, $t_4 \geq t_3$, and thus $t_5\geq t_3\geq t_1$.
As we can see, the time ordering reflects enforcement of data dependencies.
Thus, a natural way to extend WMM to \wmmDep{} is to attach a timestamp to each value, which will, in turn, impose additional constraints on rule firing in WMM.
Now we explain how to add timestamps to WMM to get \wmmDep{}.


Let us assume there is a global clock which is incremented every time a store writes the monolithic memory.
We attach a timestamp to each value in WMM, \ie{} an architecture register value, the $\langle \mathrm{address,value} \rangle$ pair of a store, and a monolithic memory value.
The timestamp represents when the value is created.
Consider an instruction $r_3=r_1+r_2$. 
The timestamp of the new value in $r_3$ will be the maximum timestamp of $r_1$ and $r_2$.
Similarly, the timestamp of the $\langle \mathrm{address,value} \rangle$ pair of a store ($\StInst\ a\ v$), \ie{} the \emph{creation} time of the store, is the maximum timestamp of all source operands to compute $\langle a,v \rangle$.
The timestamp of a monolithic memory value is the time when the value becomes visible in memory, \ie{} one plus the time when the value is stored.

Next consider a load $L$ ($\LdInst\ a$) on processor $i$, which reads the value of a store $S$ ($\StInst\ a\ v$).
No matter how WMM executes $L$ (\eg{} by reading $sb$, memory, or $ib$), the timestamp $ts$ of the load value (\ie{} the timestamp for the destination register) is always the maximum of (1) the timestamp $ats$ of the address operand, (2) the time $rts$ when processor $i$ executes the last $\RecInst$ fence, and (3) the time $vts$ when $S$ becomes visible to processor $i$.
Both $ats$ and $rts$ are straightforward.
As for $vts$, if $S$ is from another processor $j$ ($j\neq i$), then $S$ is visible after it writes memory, so $vts$ is timestamp of the monolithic memory value written by $S$.
Otherwise, $S$ is visible to processor $i$ after it is created, so $vts$ is the creation time of $S$.

A constraint for $L$, which we refer to as \emph{stale-timing}, is that $ts$ must be $\le$ the time $ts_E$ when $S$ is overwritten in memory.
This constraint is only relevant when $L$ reads from $ib$.
This constraint is needed because a load cannot return a value in \robDepccm{} if the value has been overwritten in \ccmModelName{} at the time of load execution.

To carry out the above timestamp calculus for load $L$ in WMM, we need to associate the monolithic memory $m[a]$ with the creation time of $S$ and the processor that created $S$, when $S$ updates $m[a]$.
When $S$ is overwritten and its $\langle a,v \rangle$ is inserted into $ps[i].ib$, we need to attach the time interval $[vts, ts_E]$ (\ie{} the duration that $S$ is visible to processor $i$) to that $\langle a,v \rangle$ in $ps[i].ib$.

It should be noted that PC should never be involved in the timestamp mechanism of \wmmDep{}, because the PC of each instruction can be known in advance due to the branch predictor $bp$ in \robDepModelName.

By combining the above timestamp mechanism with the original WMM rules, we have derived \wmmDep.

\subsection{\wmmDep{} Operational Semantics}

Figure \ref{fig: wmm dep op model} shows the operational semantics of \wmmDep{}.
We list the things one should remember before reading the rules in the figure.
\begin{itemize}
	\item The global clock name is $gts$ (initialized as 0), which is incremented when the monolithic memory is updated.
	
	\item Each register has a timestamp (initialized as 0) which indicates when the register value was created.
	
	\item Each $sb$ entry $\langle a,v\rangle$ has a timestamp, \ie{} the creation time of the store that made the entry.
	Timestamps are added to the method calls on $sb$ as appropriate.
	
	\item Each monolithic memory location $m[a]$ is a tuple $\langle v, \langle i, sts \rangle, mts \rangle$ (initialized as $\langle 0, \langle \Hyphen, 0 \rangle, 0 \rangle$), in which $v$ is the memory value, $i$ is the processor that writes the value, $sts$ is the creation time of the store that writes the value, and $mts$ is the timestamp of the memory value (\ie{} one plus the time of memory write). 
	
	\item Each $ib$ entry $\langle a,v \rangle$ has a time interval $[ts_L, ts_U]$, in which $ts_L$ is the time when $\langle a,v \rangle$ becomes visible to the processor of $ib$, and $ts_U$ is the time when $\langle a,v \rangle$ is overwritten in memory and gets inserted into $ib$.
	Thus, the $\insertFunc$ method on $ib$ takes the time interval as an additional argument.
	
	\item Each processor $ps[i]$ has a timestamp $rts$ (initialized as 0), which records when the latest $\RecInst$ was executed by $ps[i]$.
\end{itemize}

Some of the timestamp manipulation is done inside the decode and execute methods of each processor $ps[i]$.
Therefore we define the following methods:
\begin{itemize}
	\item $\decodeTSFunc()$: returns a pair $\langle dIns, ts \rangle$, in which $dIns$ is the decoded instruction returned by the original $\decodeFunc()$ method, and $ts$ is the maximum timestamp of all source registers (excluding PC) of $dIns$.
	\item $\executeTSFunc(dIns, ldRes, ts)$: first calls the original method $\executeFunc(dIns, ldRes)$, and then writes timestamp $ts$ to the destination register of instruction $dIns$.
\end{itemize}
We also replace the  $\getRandAndRemoveFunc$ method on $ib$ with the following two methods to facilitate the check of the stale-timing constraint:
\begin{itemize}
	\item $\getRandFunc(a)$: returns the $\langle \mathrm{value, time\ interval} \rangle$ pair of any stale value for address $a$ in $ib$.
	If $ib$ does not contain any stale value for $a$, $\langle \epsilon, \Hyphen \rangle$ is returned.
	
	\item $\removeOlderFunc(a, ts)$: removes all stale values for address $a$, which are inserted into $ib$ when $gts < ts$, from $ib$.
\end{itemize}

\begin{figure}[!htb]
	\centering
	\small
	\begin{boxedminipage}{\columnwidth}
		\textbf{\wmmDepNmRule{} rule} ($\NmInst$ execution). \reduceRuleSpace
		\begin{displaymath}
		\frac{
			\langle \langle \NmInst, dst, v \rangle, ts \rangle = ps[i].\decodeTSFunc();
		}{
			ps[i].\executeTSFunc(\langle \NmInst, dst, v \rangle, \Hyphen, ts);
		}
		\end{displaymath} \reduceRuleEndSpace
		
		\textbf{\wmmDepLdSbRule{} rule} ($\LdInst$ execution: bypass from store). \reduceRuleSpace
		\begin{displaymath}
		\frac{
			\begin{array}{c}
			\langle \langle \LdInst, a, dst \rangle, ats \rangle = ps[i].\decodeTSFunc(); \\ 
			\whenFunc(ps[i].sb.\existFunc(a));\ \langle v, sts\rangle = ps[i].sb.\getYoungestFunc(a); \\
			\end{array}
		}{
			ps[i].\executeTSFunc(\langle \LdInst, a, dst\rangle, v, \maxFunc(ats, ps[i].rts, sts));
		}
		\end{displaymath} \reduceRuleEndSpace
		
		\textbf{\wmmDepLdMemRule{} rule} ($\LdInst$ execution: read memory). \reduceRuleSpace
		\begin{displaymath}
		\frac{
			\begin{array}{c}
			\langle \langle \LdInst, a, dst \rangle, ats \rangle = ps[i].\decodeTSFunc();\ \whenFunc(\neg ps[i].sb.\existFunc(a)); \\
			\langle v, \langle j, sts \rangle, mts\rangle = m[a];\ vts = (\ifFunc\ i \neq j\ \thenFunc\ mts\ \elseFunc\ sts); \\
			\end{array}
		}{
			\begin{array}{c}
			ps[i].\executeTSFunc(\langle \LdInst, a, dst\rangle, v, \maxFunc(ats, ps[i].rts, vts)); \\ 
			ps[i].ib.\removeAddrFunc(a); \\
			\end{array}
		}
		\end{displaymath} \reduceRuleEndSpace
		
		\textbf{\wmmDepLdIbRule{} rule} ($\LdInst$ execution: read stale value). \reduceRuleSpace
		\begin{displaymath}
		\frac{
			\begin{array}{c}
			\langle \langle \LdInst, a, dst \rangle, ats \rangle = ps[i].\decodeTSFunc(); \\ 
			\langle v, [ts_L, ts_U]\rangle = ps[i].ib.\getRandFunc(a); \\
			\whenFunc(\neg ps[i].sb.\existFunc(a)\ \wedge\ v\neq \epsilon\ \wedge\ ats \leq ts_U); \\
			\end{array}
		}{
			\begin{array}{c}
			ps[i].\executeTSFunc(\langle \LdInst, a, dst\rangle, v, \maxFunc(ats, ps[i].rts, ts_L)); \\ 
			ps[i].ib.\removeOlderFunc(a, ts_U); \\
			\end{array}
		}
		\end{displaymath} \reduceRuleEndSpace
		
		\textbf{\wmmDepStRule{} rule} ($\StInst$ execution). \reduceRuleSpace
		\begin{displaymath}
		\frac{
			\langle \langle \StInst, a, v \rangle, ts \rangle = ps[i].\decodeTSFunc();
		}{
			\begin{array}{c}
			ps[i].\executeTSFunc(\langle \StInst, a, v\rangle, \Hyphen, \Hyphen); \\
			ps[i].sb.\enqFunc(a, v, ts);\ ps[i].ib.\removeAddrFunc(a);
			\end{array}
		}
		\end{displaymath} \reduceRuleEndSpace
		
		\textbf{\wmmDepRecRule{} rule} ($\RecInst$ execution). \reduceRuleSpace
		\begin{displaymath}
		\frac{
			\langle \langle \RecInst \rangle, ts \rangle = ps[i].\decodeTSFunc();
		}{
			ps[i].\executeTSFunc(\langle \RecInst \rangle, \Hyphen, \Hyphen);\ ps[i].ib.\clearFunc();\ ps[i].rts \assignVal gts;
		}
		\end{displaymath} \reduceRuleEndSpace
		
		\textbf{\wmmDepComRule{} rule} ($\ComInst$ execution). \reduceRuleSpace
		\begin{displaymath}
		\frac{
			\langle \langle \ComInst \rangle, ts \rangle = ps[i].\decodeTSFunc();\ \whenFunc(ps[i].sb.\emptyFunc());
		}{
			ps[i].\executeTSFunc(\langle \ComInst \rangle, \Hyphen, \Hyphen);
		}
		\end{displaymath} \reduceRuleEndSpace
		
		\textbf{\wmmDepDeqSbRule{} rule} (dequeue store buffer). \reduceRuleSpace
		\begin{displaymath}
		\hspace{-1pt}\frac{
			\begin{array}{c}
			a = ps[i].sb.\getAnyAddr();\ \langle v', \langle i', sts' \rangle, mts \rangle = m[a]; \\
			ts_U = gts;\ \whenFunc(a \neq \epsilon); \\
			\end{array}
		}{
			\hspace{-3pt}\begin{array}{ll}
			\multicolumn{2}{c}{\langle v, sts \rangle \leftarrow ps[i].sb.\removeOldestFunc(a);} \\ 
			\multicolumn{2}{c}{m[a] \assignVal \langle v, \langle i, sts \rangle, gts+1 \rangle;\ gts \assignVal gts + 1;} \\
			\forall j \neq i. & \hspace{-7pt} \letFunc\ ts_L = (\ifFunc\ j \neq i'\ \thenFunc\ mts\ \elseFunc\ sts')\ \inFunc \\ 
			                  & \hspace{-7pt} \ifFunc\ \neg ps[j].sb.\existFunc(a)\ \thenFunc\ ps[j].ib.\insertFunc(a, v', [ts_L, ts_U]); \\
			\end{array} \hspace{-2pt}
		}
		\end{displaymath}
	\end{boxedminipage}
	\nocaptionrule\caption{\wmmDep{} operational semantics} \label{fig: wmm dep op model}
\end{figure}

In Figure \ref{fig: wmm dep op model}, \wmmDepNmRule{} and \wmmDepStRule{} compute the timestamps of  a $\NmInst$ instruction result and a store $\langle a,v \rangle$ pair from the timestamps of source registers respectively.
\wmmDepRecRule{} updates $ps[i].rts$ with the current time because a $\RecInst$ is executed.
\wmmDepDeqSbRule{} attaches the appropriate time interval to the stale value inserted into $ib$ as described in Section \ref{sec: wmm dep intuition}.

In all three load execution rules (\wmmDepLdSbRule, \wmmDepLdMemRule, and \wmmDepLdIbRule), the timestamp of the load result is $\ge$ the timestamp of the address operand ($ats$) and the latest $\RecInst$ execution time ($ps[i].rts$).
Besides, the timestamp of the load result is also lower-bounded by the beginning time that the value is readable by the processor of the load ($ps[i]$),
In \wmmDepLdSbRule{} and \wmmDepLdIbRule{}, this beginning time (\ie{} $sts$ or $ts_L$) is stored in the $sb$ or $ib$ entry; while in \wmmDepLdMemRule, this beginning time is one of the two times (\ie{} $sts$ and $mts$) stored in the monolithic memory location depending on whether the memory value $v$ is written by $ps[i]$ (\ie{} whether $i$ is equal to $j$).
In \wmmDepLdIbRule, the stale-timing constraint requires that $\maxFunc(ats, ps[i].rts, ts_L)$ (\ie{} the timestamp of the load value) is no greater than $ts_U$ (\ie{} the time when the stale value is overwritten).
Here we only compare $ats$ with $ts_U$, because $ts_L\leq ts_U$ is obvious, and the clearing of $ib$ done by $\RecInst$ fences already ensures $ps[i].rts \leq ts_U$.

\subsection{Litmus Tests for \wmmDep} \label{sec: wmm dep litmus}

\noindent\textbf{Enforcing data dependency:}
First revisit the behavior in Figure \ref{fig: wmm val pred}.
In \wmmDep, the timestamp of the source operand of $I_5$ (\ie{} the result of $I_4$) is 2, while the time interval of the stale value 0 for address $a$ in the $ib$ of P1 is $[0,0]$.
Thus $I_5$ cannot read 0 from $ib$, and the behavior is forbidden.
For a similar reason, \wmmDep{} forbids the behavior in Figure \ref{fig: wmm dep st to ld}, in which $I_4$ carries data dependency to $I_7$ transitively.
This behavior is also impossible in \robDepccm.

\noindent\textbf{Allowing other speculations:}
\wmmDep{} still allows the behavior in Figure \ref{fig: wmm mem dep pred} which can result from memory dependency speculation in hardware.
As we can see, \wmmDep{} still allows implementations to speculate on all dependencies other than data dependency.

\begin{figure}[!htb]
	\centering
	\begin{minipage}[b]{0.46\columnwidth}
		\centering
		\small
		\begin{tabular}{|l|l|}
			\hline
			Proc. P1 & Proc. P2 \\
			\hline
			$\!\! I_1: \StInst\ a\ 1 \!\!$ & $\!\! I_4: r_1=\LdInst\ b \!\!$ \\
			$\!\! I_2: \ComInst \!\!$      & $\!\! I_5: \StInst\ c\ r_1 \!\!$ \\
			$\!\! I_3: \StInst\ b\ a \!\!$ & $\!\! I_6: r_2=\LdInst\ c \!\!$ \\
			                               & $\!\! I_7: r_3 = \LdInst\ r_2\!\!$ \\
			\hline
			\multicolumn{2}{|l|}{\wmmDep{} forbids: $r_1=a,$} \\
			\multicolumn{2}{|l|}{$r_2=a, r_3=0$} \\
			\hline
		\end{tabular}
		\nocaptionrule \caption{Transitive data dependency} \label{fig: wmm dep st to ld}
	\end{minipage}
	\hspace{3pt}
	\begin{minipage}[b]{0.50\columnwidth}
		\small
		\centering
		\begin{tabular}{|l|l|}
			\hline
			Proc. P1 & Proc. P2 \\
			\hline
			$\!\! I_1: \StInst\ a\ 1 \!\!$ & $\!\! I_4: r_1=\LdInst\ b \!\!$ \\
			$\!\! I_2: \ComInst \!\!$      & $\!\! I_5: r_2 \!=\! r_1 \!+\! c \!-\! 1 \!\!$ \\
			$\!\! I_3: \StInst\ b\ 1 \!\!$ & $\!\! I_6: r_3=\LdInst\ r_2 \!\!$ \\
			                               & $\!\! I_7: r_4=\LdInst\ c \!\!$ \\
			                               & $\!\! I_8: r_5 = r_4 \!+\! a \!\!$ \\
			                               & $\!\! I_9: r_6=\LdInst\ r_5 \!\!$ \\
			\hline
			\multicolumn{2}{|l|}{\hspace{-3pt}\wmmDep{} allows: $r_1=1, r_2=c \!\!$} \\
			\multicolumn{2}{|l|}{$r_3=0, r_4=0, r_5=a, r_6=0$} \\
			\hline
		\end{tabular}
		\nocaptionrule \caption{RSW in \wmmDep} \label{fig: rsw}
	\end{minipage}
\end{figure}

\noindent\textbf{Loads to the same address:}
Remember that two loads to the same address can be executed out of order in \robDepModelName{} as long as they read from the same store.
\wmmDep{} captures this subtle optimization.
Consider the Read-from-Same-Write (RSW) program in Figure \ref{fig: rsw}.
The behavior is observable in \robDepccm, because $I_7$ to $I_9$ can be executed before $I_4$ to $I_6$.
It is fine for $I_6$ and $I_7$, which read the same address $c$, to be executed out-of-order, because they both read from the initialization store.
\wmmDep{} allows this behavior, because the timestamp of the address operand of $I_9$ is 0, and $I_9$ can read stale value 0 from $ib$.
(This behavior is also observable on POWER and ARM processors \cite{sarkar2011understanding,flur2016modelling}).

\subsection{Equivalence of \wmmDep{} and \robDepccm}

To simplify the proof, we change the $\findRecStFunc(idx)$ method in \robDepModelName{} to return $\epsilon$ whenever there is a $\RecInst$ fence at index smaller than $idx$ in ROB, \ie{} a load will always be stalled when there is an older $\RecInst$ in ROB.
This change in $\findRecStFunc$ only affects one scenario: a load used to be able to bypass from a store when there is a $\RecInst$ fence older than both the load and store in ROB, and operations dependent on the load result could used to be done before the $\RecInst$ is committed from ROB.
Since the bypass and those dependent operations are all local to the processor, they can still be performed with the same effects immediately after the $\RecInst$ is committed from ROB.
Thus, the change in $\findRecStFunc$ does not affect the semantics of \robDepModelName.

\begin{theorem}
	\robDepccm{} $\subseteq$ \wmmDep.
\end{theorem}
\begin{proof}
	We also introduce the global clock $gts$ to \robDepccm; $gts$ is be incremented whenever the \ccmStRule{} fires.
	This proof is almost the same as that of Theorem \ref{thm: rob < wmm} except for a new invariant: The timestamp computed for each instruction result (\ie{} $\NmInst$ result, store address and data, and load value) in \wmmDep{} is $\le$ the the value of $gts$ when the instruction gets its result in \robDepccm{}.
	See Appendix~\ref{sec:sup-wmm-d} for the complete proof.
\end{proof}

\begin{theorem}
	\wmmDep{} $\subseteq$ \robDepccm.
\end{theorem}
\begin{proof}
	We also introduce the global clock $gts$ to \robDepccm, $gts$ is be incremented whenever the \ccmStRule{} fires.
	This proof is similar to the proof for Theorem \ref{thm: rob > wmm} except for the following two points:
	\begin{enumerate}
		\item Without a value predictor in \robDepModelName, the time when an instruction can be executed in \robDepccm{} is subject to when the source operands of the instruction become ready.
		When constructing the execution for \robDepccm{} (to simulate the behavior of \wmmDep), we always fire the instruction execution rule as early as possible.
		In particular for each $\LdInst$ instruction, we execute it (by firing the \robLdAddrRule{} and \robLdBypassRule{} rules or the \robLdAddrRule{}, \robLdReqRule, and \ccmLdRule{} rules) as soon as the source operands are ready and the expected load value (\ie{} the value read in \wmmDep{}) becomes visible to the processor.
		
		\item A new invariant: the value of $gts$ in \robDepccm{} when an instruction gets its result ($\NmInst$ result, store address and data, or load value) is equal to the timestamp of that result in \wmmDep.
	\end{enumerate}
	See Appendix~\ref{sec:sup-wmm-d} for the complete proof.
\end{proof}

\section{Modeling Multi-Copy Non-Atomic Stores} \label{sec: non atomic mem}

Unlike the multi-copy atomic stores in WMM, stores in ARM and POWER multiprocessors are multi-copy non-atomic, \ie{} a store may become visible to different processors at different times.
This is caused by sharing store buffers or write-through caches in the memory system.
If multiple threads share a store buffer or a write-through cache, a store by any of these threads may be seen by all these threads before other processors.
Although we could tag stores with thread IDs in the store buffer, it is infeasible to distinguish between stores by different threads in the write-through cache.
While \ccmModelName{} cannot model such store behaviors, the storage subsystem of the Flowing model is believed to have captured precisely the behaviors of this multi-copy non-atomicity given a topology of the hierarchy of shared store buffers or write-through caches \cite{flur2016modelling}.

In this section, we first introduce a new \IIE{} model, \wmmSSB, which captures the multi-copy non-atomic store behaviors in a topology-independent way.
\wmmSSB{} is derived from WMM by changing the store buffers to a new conceptual device: \emph{dynamic store buffers}.
Next we introduce \robSSBflow, the physical model for multiprocessors with multi-copy non-atomic stores; \flowModelName{} is the memory abstraction taken from the Flowing model and \robSSBModelName{} is the processor model adapted from \robModelName{}.
We will finally prove \robSSBflow{} $\subseteq$ \wmmSSB.

\subsection{\wmmSSB{}: Copying From One Store Buffer into Another} \label{sec: st propagate}

We can model the multi-copy non-atomicity of stores by introducing a background rule to make copies of a store in a store buffer into other store buffers.
We refer to these store buffers with the ability of copying stores as \emph{dynamic store buffers}.
(We will still use store buffers to refer to dynamic store buffers in the rest of this section).
However, we need to ensure that all stores for an address can be put in a total order, \ie{} the coherent order ($\coEdge$), and the order seen by any processor is consistent with this total order (\ie{} SC for a single address).
{\wmmSSB{}} is an \IIE{} model to generate such behaviors.

To identify all the copies of a store in various store buffers, we assign a unique tag $t$ when a store is inserted in the store buffer, and this tag is copied when a store is copied from one store buffer to another.
When a store is committed from the store buffer to the memory, all its copies must be deleted from all the store buffers which have them. 
A store can be committed only if all its copies are the oldest store for that address in their respective store buffers.

All the stores for an address in a store buffer are kept as a strictly ordered list, where the \emph{youngest} store is the one that entered the store buffer \emph{last}. 
We make sure that all ordered lists are can be combined transitively to form a strict partial order, which has now to be understood in terms of the tags on stores because of the copies.
By the end of the program, this partial order on the stores for an address becomes the coherence order, so we refer to this partial order as the \emph{partial coherence order}.

Consider the states of store buffers shown in Figure \ref{fig: store propagate example}.
$A$, $B$, $C$ and $D$ are different stores to the same address, and their tags are $t_A$, $t_B$, $t_C$ and $t_D$, respectively.
$A'$ and $B'$ are copies of $A$ and $B$ respectively created by the background copy rule.
Ignoring $C'$, the partial coherence order contains: $t_D \coOrd t_B \coOrd t_A$ ($D$ is older than $B$, and $B$ is older than $A'$ in P2), and $t_C \coOrd t_B$ ($C$ is older than $B'$ in P3).
Note that $t_D$ and $t_C$ are not related here.

At this point, if we copied $C$ in P3 as $C'$ into P1, we would add a new edge $t_A \coOrd t_C$, which would break the partial order by introducing the cycle $t_A \coOrd t_C \coOrd t_B \coOrd t_A$.
Therefore copying of $C$ into P1 should not be allowed in this state.
Similarly, copying a store with tag $t_A$ into P1 or P2 should be forbidden because it would immediately create a cycle:  $t_A \coOrd t_A$.
In general, the background copy rule must be constrained so that invariance of the partial coherence order after copying is maintained.

\begin{figure}[!htb]
	\centering
	\includegraphics[width=0.75\columnwidth]{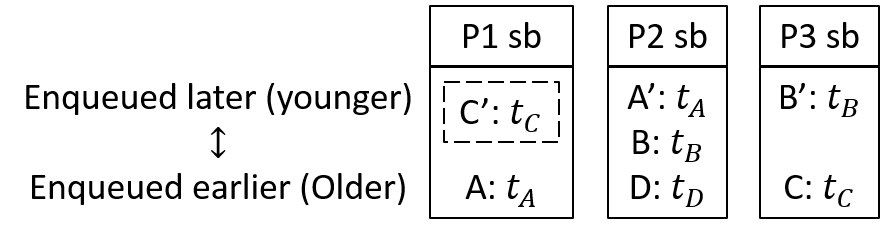}
	\nocaptionrule\caption{Example states of store buffers (primes are copies)} \label{fig: store propagate example}
\end{figure}

The operational semantics of \wmmSSB{} is defined by adding/replacing three rules to that of WMM (Figure \ref{fig: wmm rule}).
These new rules are shown in Figure \ref{fig: wmm ssb op model}: 
A new background rule \wmmSSBPropSt{} is added to WMM and the \wmmSSBStRule{} and \wmmSSBDeqSbRule{} rules replace the \wmmStRule{} and \wmmDeqSbRule{} rules of WMM, respectively.
Before reading these new rules, one should note the following facts:
\begin{itemize}
	\item The $\decodeFunc$ method now returns $\langle \StInst, a, v, t\rangle$ for a store, in which $t$ is the unique tag assigned to the store.
	Each store buffer entry becomes a tuple $\langle a,v,t \rangle$, in which $t$ is the tag.
	Tags are also introduced into the methods of $sb$ appropriately.
	
	\item The $sb$ now has the following three methods:
	\begin{itemize}
		\item $\hasTagFunc(t)$: returns $\True$ if $sb$ contains a store with tag $t$.
		\item $\getOldestFunc(a)$: returns the $\langle \mathrm{value, tag} \rangle$ pair of the oldest store to address $a$ in $sb$.
		It returns $\langle \epsilon,\epsilon\rangle$ if $sb$ does not contain $a$.
		\item $\getRandFunc()$: returns the $\langle \mathrm{address, value, tag} \rangle$ tuple of any store present in $sb$.
		It returns $\langle \epsilon,\epsilon,\epsilon\rangle$ if $sb$ is empty.
	\end{itemize}
	
	\item A new function $\noCycleFunc(a, t, j)$ is defined to check whether the background rule could copy a store with tag $t$ for address $a$ into the $sb$ of processor $j$.
	It returns $\True$ if the partial coherence order among the tags of all stores for address $a$ does not contain any cycle after doing the copy.
\end{itemize}

\begin{figure}[!htb]
	\centering
	\begin{boxedminipage}{\columnwidth}
		\small
		\textbf{\wmmSSBStRule{} rule} ($\StInst$ execution). \reduceRuleSpace
		\begin{displaymath}
		\frac{
			\langle \StInst, a, v, t \rangle = ps[i].\decodeFunc();
		}{
			\begin{array}{c}
			ps[i].\executeFunc(\langle \StInst, a, v, t \rangle, \Hyphen); \\
			ps[i].sb.\enqFunc(a, v, t);\ ps[i].ib.\removeAddrFunc(a); \\
			\end{array}
		}
		\end{displaymath} \reduceRuleEndSpace
		
		\textbf{\wmmSSBDeqSbRule{} rule} (dequeue store buffer). \reduceRuleSpace
		\begin{displaymath}
		\frac{
			\begin{array}{c}
			a = ps[i].sb.\getAnyAddr();\ old = m[a];\ \langle v,t \rangle = ps[i].sb.\getOldestFunc(a); \\
			\multicolumn{1}{l}{\whenFunc(a \neq \epsilon\ \wedge} \\
			\multicolumn{1}{r}{\forall j.\ \neg ps[j].sb.\hasTagFunc(t)\ \vee\ ps[j].sb.\getOldestFunc(a) == \langle v,t \rangle);} \\
			\end{array}
		}{
			\begin{array}{ll}
			\multicolumn{2}{c}{m[a] \assignVal v;} \\
			\forall j. & \ifFunc\ ps[j].sb.\hasTagFunc(t)\ \thenFunc\ ps[j].sb.\removeOldestFunc(a); \\ 
			& \elseFunc\ \ifFunc\ \neg ps[j].sb.\existFunc(a)\ \thenFunc\ ps[j].ib.\insertFunc(a,old); \\
			\end{array}
		}
		\end{displaymath} \reduceRuleEndSpace
		
		\textbf{\wmmSSBPropSt{} rule} (copy store from processor $i$ to $j$). \reduceRuleSpace
		\begin{displaymath}
		\frac{
			\begin{array}{c}
				\langle a, v, t \rangle = ps[i].sb.\getRandFunc();\ \whenFunc(a\neq \epsilon\ \wedge\ \noCycleFunc(a, t, j)); \\
			\end{array}
		}{
			ps[j].sb.\enqFunc(a,v,t);\ ps[j].ib.\removeAddrFunc(a);
		}
		\end{displaymath}
	\end{boxedminipage}
	\nocaptionrule\caption{\wmmSSB{} operational semantics} \label{fig: wmm ssb op model}
\end{figure}

In Figure \ref{fig: wmm ssb op model}, \wmmSSBStRule{} simply introduces the store tag to the original \wmmStRule{} rule.
In \wmmSSBDeqSbRule{}, when we write a store ($\langle a, v, t\rangle$) into memory, we ensure that each copy of this store is the oldest one to address $a$ in its respective $sb$.
The old memory value is inserted into the $ib$ of each processor whose $sb$ does not contain address $a$.
\wmmSSBPropSt{} copies a store ($\langle a,v,t \rangle$) from $ps[i]$ to $ps[j]$.
The check on $\noCycleFunc(a,t,j)$ guarantees that no cycle is formed in the partial coherence order after the copy.
Copying stores from $ps[i]$ to $ps[i]$ will be automatically rejected because $\noCycleFunc$ will return $\False$.
Since we enqueue a store into $ps[j].sb$, we need to remove all stale values for address $a$ from $ps[j].ib$.

\wmmSSB{} still use the \wmmComRule{} to execute a $\ComInst$ fence, but this rules has very different implications here.
In \wmmSSB{}, a store cannot be moved from $sb$ to memory unless all its copies in other store buffers can be moved at the same time.
Hence the effect of a $\ComInst$ fence is not local; it implicitly affects all other store buffers.
In literature, such fences are known as \emph{cumulative}.

\subsection{Litmus Tests for \wmmSSB}

We show by examples that \wmmSSB{} allows multi-copy non-atomic store behaviors, and that fences in \wmmSSB{} have the cumulative properties similar to those in POWER and ARM memory models.

We first consider the Write-Write Causality (WWC) example in Figure \ref{fig: wwc wmm ssb} (which is forbidden by WMM).
\wmmSSB{} allows this behavior by first copying $I_1$ into the $sb$ of P2 to let all instructions on P2 and P3 proceed.
$I_1$ will write memory only after $I_5$ has done so.
This behavior is allowed in hardware in case a store buffer is shared by P1 and P2 but not P3.
To forbid this behavior in \wmmSSB, a $\ComInst$ fence is required between $I_2$ and $I_3$ on P2 to push $I_1$ into memory.
The inserted $\ComInst$ fence has a cumulative global effect in ordering $I_1$ before $I_3$ (and hence $I_5$).

Figure \ref{fig: iriw wmm ssb} shows another well-known example called Independent Reads of Independent Writes (IRIW), which is forbidden by WMM.
\wmmSSB{} allows this by copying $I_1$ and $I_2$ into the $sb$ of P3 and P4 respectively.
This behavior is possible in hardware, in case P1 and P3 shares a store buffer while P2 and P4 shares a different one.

To forbid the behavior in Figure \ref{fig: iriw wmm ssb} in \wmmSSB, we can insert a $\ComInst$ fence between $I_3$ and $I_4$ on P3, and another $\ComInst$ fence between $I_6$ and $I_7$ on P4.
As we can see, a $\ComInst$ followed by a $\RecInst$ in \wmmSSB{} has a similar effect as the POWER $\syncInst$ fence and the ARM $\dmbInst$ fence.
Cumulation is achieved by globally advertising observed stores ($\ComInst$) and preventing later loads from reading stale values ($\RecInst$).

\begin{figure}[!htb]
	\centering
	\begin{minipage}{\columnwidth}
		\centering
		\small
		\begin{tabular}{|l|l|l|}
			\hline
			Proc. P1 & Proc. P2 & Proc. P3 \\
			\hline
			$I_1: \StInst\ a\ 2$ & $I_2: r_1=\LdInst\ a$        & $I_4: r_2=\LdInst\ b$     \\
			                     & $I_3: \StInst\ b\ (r_1 - 1)$ & $I_5: \StInst\ a\ r_2$     \\
			\hline
			\multicolumn{3}{|l|}{RC forbids: $r_1=2,\ r_2=1,\ m[a]=2$} \\
			\hline
		\end{tabular}
		\nocaptionrule \caption{WWC in \wmmSSB} \label{fig: wwc wmm ssb}
	\end{minipage}\\
	\vspace{8pt}
	\begin{minipage}{\columnwidth}
		\begin{tabular}{|l|l|l|l|}
			\hline
			Proc. P1 & Proc. P2 & Proc. P3 & Proc. P4 \\
			\hline
			$I_1: \StInst\ a\ 1$ & $I_2: \StInst\ b\ 1$ & $I_3: r_1=\LdInst\ a$ & $I_6: r_3=\LdInst\ b$ \\
			                     &                      & $I_4: \RecInst$       & $I_7: \RecInst$ \\
			                     &                      & $I_5: r_2=\LdInst\ b$ & $I_8: r_4=\LdInst\ a$ \\
			\hline
			\multicolumn{4}{|l|}{\wmmSSB{} allows: $r_1=1,\ r_2=0,\ r_3=1,\ r_4=0$} \\
			\hline
		\end{tabular}
		\nocaptionrule \caption{IRIW in \wmmSSB} \label{fig: iriw wmm ssb}
	\end{minipage}
\end{figure}

\subsection{\robSSBflow: a Physical Model for Muiltiprocessors with Multi-Copy Non-Atomic Memory}

For clarity of discussion, we reiterate the important concepts of the original Flowing model (see Section 7.1$\sim$7.4 in \cite{flur2016modelling}) while describing the changes made for \flowModelName{} due the differences in the fences of \wmmSSB{} and ARM.
\flowModelName{} has $n$ ports, each of which will be connected to a processor.
In addition to the $\reqLdFunc(idx, a)$ and $\reqStFunc(a,v,t)$ port methods ($idx$ is the ROB index, \ie{} the load request tag, and $t$ is the unique store tag), \flowModelName{} provides a $\reqComFunc()$ port method so that the processor could send a $\ComInst$ fence (instead of the ARM $\dmbInst$ fence) as a barrier request into the memory system.
Inside \flowModelName, there are $k$ segments $s[1\ldots k]$, and a monolithic memory $m$ (same as the one in \ccmModelName).
All the segments and the monolithic memory are connected together into a tree rooted at $m$.
Each segment contains a list of memory requests (\ie{} loads, stores and $\ComInst$).
Each port of \flowModelName{} is connected to a segment in the tree, and the port methods (\ie{} $\reqLdFunc$, $\reqStFunc$ and $\reqComFunc$) simply add the new request to the top of the list of the connected segment.
\flowModelName{} has the following three internal rules:
\begin{enumerate}
	\item \textbf{\flowReorderRule:} Two consecutive requests $r_{new}$ and $r_{old}$ in the same segment ($r_{new}$ is closer to the top of the list of the segment) can be reordered except for the following two cases:
	\begin{enumerate}
		\item $r_{new}$ and $r_{old}$ are memory accesses to the same address.
		\item $r_{new}$ is a $\ComInst$ and $r_{old}$ is a store.
	\end{enumerate}
	\item \textbf{\flowBypassRule:} When a load request $r=\langle \LdInst, pid, idx, a\rangle$ (\ie{} a load to address $a$ from the ROB entry with index $idx$ of processor $pid$) and a store request $r'=\langle \StInst, a, v, t\rangle$ (\ie{} a store to address $a$ with data $v$ and tag $t$) are two consecutive requests in the segment and $r'$ is closer to the bottom of the segment, we can remove $r$ from the segment and call method $\respLdFunc(idx, v, t)$ of processor $pid$ (\ie{} $r$ is satisfied by $r'$).
	\item \textbf{\flowFlowRule:} The request $r$ at the bottom of segment $s[i]$ can be removed from $s[i]$.
	If the parent of $s[i]$ is another segment $s[j]$, we add $r$ to the top of $s[j]$ (\ie{} $r$ flows from $s[i]$ to its parent $s[j]$).
	Otherwise, the parent of $s[i]$ is $m$, and we take the following actions according to the type of $r$:
	\begin{itemize}
		\item If $r$ is $\langle \LdInst, pid, idx, a \rangle$, we call method $\respLdFunc(idx, v, t)$ of processor $pid$, in which pair $\langle v, t\rangle$ is the current state of $m[a]$. 
		\item If $r$ is $\langle \ComInst, pid \rangle $, \ie{} a $\ComInst$ fence from processor $pid$, we call method $\respComFunc()$ (which is defined later) of processor $pid$ to indicate the completion of the fence.
		\item If $r$ is $\langle \StInst, a, v, t\rangle$, we update $m[a]$ to be $\langle v, t\rangle$.
		No response is sent for the store request.
	\end{itemize}
\end{enumerate}

We adapt \robModelName{} to \robSSBModelName{} to fit the new memory system.
The operational semantics (the changed part) and the interface methods of \robSSBModelName{} are shown in Figures \ref{fig: rob ssb rule} and \ref{fig: rob ssb interface} respectively, where $\fm$ represents the \flowModelName{} interface port connected to the processor, and method $rob.\setCommitFunc(en)$ sets the commit slot of ROB to $en$.

The first change for \robSSBModelName{} is to remove $sb$ from each processor, because store buffering is already modeled inside \flowModelName.
Thus when a store is committed from ROB in rule \robSSBStRetRule{}, the store request is directly sent to \flowModelName.
The second change is sending a $\ComInst$ request to \flowModelName{} when a $\ComInst$ fence reaches the commit slot of ROB as shown in rule \robSSBComRetRule{}.
When the $\ComInst$ response comes back from \flowModelName{} via method $\respComFunc$, the fence is committed from ROB.
To avoid duplicate $\ComInst$ requests to \flowModelName{}, we change function $\initExFunc$ to also return $\Idle$ for a $\ComInst$ fence, and hence the $ex$ field of a $\ComInst$ fence will be set to $\Idle$ in rule \robFetchRule{}.
When the $\ComInst$ request is sent to \flowModelName{} in rule \robSSBComRetRule{}, we set the $ex$ field to $\Exe$.

The last change is about detecting whether the out-of-order execution of loads to the same address in the same processor violates SC for single address.
The detection is harder in case of \flowModelName{} than that in \ccmModelName, because loads can be satisfied in any segment or monolithic memory inside \flowModelName, while loads can only be satisfied in the monolithic memory in case of \ccmModelName.
The original Flowing model has specified complicated conditions of this check to avoid unnecessary flush of loads, but we believe those conditions may still cause some loads to be flushed unnecessarily.
Instead of further complicating the check, we simply guarantee that loads to the same address are issued to \flowModelName{} in order.
Since \flowModelName{} keeps the order of memory accesses to the same address, this can ensure SC for single address.
Rule \robSSBLdReqRule{} enforces the in-order issue by killing younger loads in the same way as the \robStExRule{} rule does.
This also makes \robSSBModelName+\ccmModelName{} obey the \descref{CoRR} and \descref{CoWR} axioms (see Appendix~\ref{sec:sup-sc} for the proof).

\begin{figure}[!htb]
	\centering
	\begin{boxedminipage}{\columnwidth}
		\small
		\textbf{\robSSBStRetRule{} rule} (commit $\StInst$ from ROB). \reduceRuleSpace
		\begin{displaymath}
		\frac{
			\langle pc, npc, ins, \Done \rangle = rob.\getCommitFunc();\ \langle \StInst, srcs, a, v, t \rangle = ins;
		}{
			rob.\deqFunc();\ \fm.\reqStFunc(a, v, t);\ \rf.\updateFunc(ins);
		}
		\end{displaymath} \reduceRuleEndSpace
		
		\textbf{\robSSBComRetRule{} rule} ($\ComInst$ issue). \reduceRuleSpace
		\begin{displaymath}
		\frac{
			\langle pc, npc, \langle \ComInst \rangle, \Idle \rangle = rob.\getCommitFunc();
		}{
			rob.\setCommitFunc(\langle pc, npc, \langle \ComInst \rangle, \Exe \rangle);\ \fm.\reqComFunc();
		}
		\end{displaymath} \reduceRuleEndSpace
		
		\textbf{\robSSBLdReqRule{} rule} ($\LdInst$ execution by sending request to \flowModelName). \reduceRuleSpace
		\begin{displaymath}
		\frac{
			\begin{array}{c}
			\langle idx, \langle pc, npc, \langle \LdInst, srcs, dst, a, v, \epsilon \rangle, \Idle \rangle \rangle = rob.\getReadyFunc(); \\
			list = rob.\findLdKilledByStFunc(idx, a);\ res = rob.\findRecStFunc(idx, a); \\ 
			\whenFunc(a \neq \epsilon\ \wedge\ res == \top\ \wedge\ \neg \lf[idx]); \\
			\end{array}
		}{
			\begin{array}{c}
			\fm.\reqLdFunc(idx, a);\ rob[idx].ins \assignVal \Exe; \\
			\begin{array}{l}
			\mathbf{for}\ kIdx = list.\mathsf{first}()\ \mathbf{to}\ list.\mathsf{tail}() \\
			\qquad \ifFunc\ rob[kIdx].ex == \Done\ \thenFunc \\
			\qquad \qquad rob.\flushFunc(kIdx, rob[kIdx].pc);\ \mathbf{break}; \\
			\qquad \elseFunc\ rob[kIdx].ex \assignVal \ReEx; \\
			\end{array} \\
			\end{array}
		}
		\end{displaymath}
	\end{boxedminipage}
	\nocaptionrule\caption{\robSSBModelName{} operational semantics} \label{fig: rob ssb rule}
	\vspace{8pt}
	\centering
	\begin{boxedminipage}{\columnwidth}
		\small
		$\respLdFunc(idx, res, t)$ \textbf{method:} \reduceRuleSpace
		\begin{displaymath}
		\begin{array}{l}
		\ifFunc\ \lf[idx]\ \thenFunc\ \lf[idx] \assignVal \False;\ \textcolor{blue}{/\!/\ \mathrm{wrong\Hyphen{}path\ load\ response}}  \\
		\elseFunc \\
		\quad \letFunc\ \langle pc, npc, \langle \LdInst, srcs, dst, a, v, \epsilon \rangle, ex \rangle = rob[idx]\ \inFunc \\
		\quad \quad \ifFunc\ ex == \ReEx\ \thenFunc\ rob[idx].ex \assignVal \Idle; \\
		\quad \quad \elseFunc\ \codecomment{/\!/\ \mathrm{save\ load\ result\ and\ check\ value\ misprediction}} \\
		\quad \quad \quad rob[idx] \assignVal \langle pc, npc, \langle \LdInst, srcs, dst, a, res, t \rangle, \Done \rangle; \\
		\quad \quad \quad \ifFunc\ v \neq \epsilon\ \wedge\ v \neq res\ \thenFunc\ rob.\flushFunc(idx+1, npc); \\
		\end{array}
		\end{displaymath} \reduceRuleEndSpace
		
		$\respComFunc()$ \textbf{method:} $rob.\deqFunc(); \rf.\updateFunc(\langle \ComInst \rangle);$
	\end{boxedminipage}
	\nocaptionrule\caption{\robSSBModelName{} interface methods} \label{fig: rob ssb interface}
\end{figure}

\subsection{\wmmSSB{} Abstracting \robSSBflow} \label{sec: wmm-s > hw}

\begin{theorem}
	\robSSBflow{} $\subseteq$ \wmmSSB.
\end{theorem}
\begin{proof}
	Similar to the proof of Theorem \ref{thm: rob < wmm}, we only consider executions in \robSSBflow{} without any ROB flush or load re-execution.
	At any moment in \robSSBflow, we define a store $S$ is \emph{observed by commits of processor $i$ ($ooo[i]$)} if and only if $S$ has been committed from $ooo[i].rob$ or $S$ has been returned by a load $L$ which has been committed from $ooo[i].rob$.
	For two stores $S_1$ and $S_2$ both observed by commits of a processor, we say $S_1$ is \emph{closer to the root} than $S_2$ in \flowModelName, when the segment of $S_1$ is closer to the root than that of $S_2$ in the tree hierarchy of \flowModelName{} (assuming each edge in the tree has length 1), or when $S_1$ and $S_2$ are in the same segment and $S_1$ is closer to the bottom of the segment.

	For any execution $E$ in \robSSBflow, we simulate it in \wmmSSB{} using the following way:
	\begin{itemize}
		\item When a store request flows into the monolithic memory inside \flowModelName{} ($\fm.m$) using the \flowFlowRule{} rule of \robSSBflow{}, \wmmSSB{} fires a \wmmSSBDeqSbRule{} rule to dequeue that store from all store buffers and write it into the monolithic memory ($m$).
		\item When a $\NmInst$, $\RecInst$, or $\StInst$ instruction is committed from ROB in \robSSBflow{}, we execute that instruction in \wmmSSB.
		\item When a $\ComInst$ fence is committed from ROB in the $\respComFunc$ method called by a \flowFlowRule{} rule in \robSSBflow{}, we execute that fence in \wmmSSB.
		\item When a $\LdInst$ instruction $L$, which reads from a store $S$, is committed from $ooo[i].rob$ in \robSSBflow{}, we execute $L$ using the following actions in \wmmSSB{} according to the status of $S$ right before the commit of $L$ in \robSSBflow:
		\begin{itemize}
			\item If $S$ is in $\fm.m$ at that time, then \wmmSSB{} executes $L$ by reading from $m$.
			\item If $S$ has been overwritten in $\fm.m$ before the commit of $L$, then $L$ can read $ps[i].ib$ in \wmmSSB{}.
			\item If $S$ is in a segment inside \flowModelName{} at that time, then $L$ should read $ps[i].sb$ in \wmmSSB{}.
			In case $S$ is not observed by commits of $ooo[i]$ right before the commit of $L$, we fire a \wmmSSBPropSt{} rule to copy $S$ into $ps[i].sb$ right before $L$ reads $ps[i].sb$.
		\end{itemize}
	\end{itemize}
	After each step of the simulation, we could prove inductively that the following invariants hold:
	\begin{enumerate}
		\item All stores in $ps[i].sb$ in \wmmSSB{} are exactly the set of stores, which are in the segments of \flowModelName{} and observed by commits of $ooo[i]$. 
		The segments that contain these stores must be on the path from the root to port $i$ in the tree hierarchy of \flowModelName{}.
		\item In case $S_1$ and $S_2$ are two stores to the same address in $ps[i].sb$, $S_1$ is older than $S_2$ in $ps[i].sb$ if and only if $S_1$ is closer to the root than $S_2$ in \flowModelName{}.
	\end{enumerate}
	We do not take any action in \wmmSSB{} when requests flow between segments in \flowModelName{} or when requests are reordered in \flowModelName{}.
	These operations in \robSSBflow{} do not affect the above invariants.

	Here we only consider the final case, \ie{} a load $L$ to address $a$, which reads from a store $S$, is committed from $ooo[i].sb$ in \robSSBflow{}, and $S$ remains in a segment of \flowModelName{} at the commit time of $L$.
	If $S$ has already been observed by commits of $ooo[i]$ before the commit of $L$, the \descref{CoRR} and \descref{CoWR} axioms of \robSSBflow{} imply that there cannot be any store to $a$ which is also observed by commits of $ooo[i]$ and is further from root than $S$ in \flowModelName{}.
	Thus, $S$ must be the youngest store to $a$ in $ps[i].sb$ and $L$ can read from it in \wmmSSB.
	Otherwise, $S$ is not observed by commits of $ooo[i]$ right before the commit of $L$, and the \descref{CoRR} and \descref{CoWR} axioms of \robSSBflow{} imply that $S$ must be further from root in the tree hierarchy of \flowModelName{} than any store observed by commits of $ooo[i]$ at that time.
	Therefore, if we insert $S$ into $ps[i].sb$, both invariants still hold.
	Since there is no cycle in the tree hierarchy of \flowModelName{} and the order of stores to the same address in any store buffer in \wmmSSB{} is the same as the order of distance from the root of those stores in \flowModelName{}, the $\noCycleFunc$ check in the \wmmSSBPropSt{} rule which copies $S$ into $ps[i].sb$ must succeed.
	Then $L$ can read from $S$ in $ps[i].sb$.
	
	See Appendix~\ref{sec:sup-wmm-s} for remaining cases.
\end{proof}

\section{Conclusion} \label{sec: conclude}

We provide a framework which uses simple hardware abstractions based on \IIE{}
processors, monolithic memory, invalidation buffers, timestamps and dynamic
store buffers to capture all microarchitectural optimizations present in modern
processors. We have proved the equivalences between the simple
abstractions and their realistic microarchitectural counterparts; we believe
this work can be useful for both programmers to reason about their programs on
real hardware, and on architects to reason about the effect of their optimizations
on program behavior.

\bibliographystyle{abbrvnat}
\bibliography{ref}

\appendix
\newtheorem*{theorem*}{Theorem}

\newtheorem{thm}{Theorem}
\setcounter{thm}{2}

\section{Proof of Equivalence Between \wmmDep{} and \robDepccm}\label{sec:sup-wmm-d}
\begin{thm}
	\robDepccm{} $\subseteq$ \wmmDep.
\end{thm}
\begin{proof}
	In order to relate the time in \robDepccm{} to that in \wmmDep, we also introduce the global clock $gts$ to \robDepccm.
	$gts$ is be incremented by one whenever the \ccmStRule{} rule fires, \ie{} when the monolithic memory of \ccmModelName{} ($ccm.m$) is written.
	In the rest of this proof, we mean the value of $gts$ when referring to time.
    Similar to the proof of Theorem 1 (\ie{} \robccm{} $\subseteq$ WMM in the paper), we only need to consider executions in \robDepccm{} without any ROB flush or load re-execution.
	
	For any execution $E$ in \robDepccm, we simulate it using the same way as in the proof of Theorem 1, \ie{} when \robDepccm{} commits an instruction from ROB or writes a store into $ccm.m$, \wmmDep{} executes that instruction or writes that store into $m$.
	After each step of simulation, we maintain the following invariants ($ooo[i]$ represents processor $i$ in \robDepccm):
	\begin{enumerate}
		\item The states of $m$ and all $sb$ in WMM are the same as $ccm.m$ and all $sb$ in \robDepccm.
		\item The $gts$ in WMM is the same as that in \robDepccm.
		\item All instructions committed from ROBs in \robDepccm{} have also been executed with the same results in WMM.
		\item The $rts$ of each processor $ps[i]$ in WMM is equal to the time when the last $\RecInst$ fence is committed from $ooo[i].rob$.
		\item The timestamp of the result of each $\NmInst$ or $\LdInst$ instruction in WMM is $\le$ the time when the instruction gets its result in \robDepccm{} by the \robNmExRule, \robLdBypassRule, or \ccmLdRule{} rule.
		\item The timestamp of each store in any $sb$ in WMM is $\le$ the time when the store is executed in \robDepccm{} by the \robStExRule{} rule.
		\item For each monolithic memory location $m[a]=\langle v, \langle i, sts \rangle, mts\rangle$ in WMM, $sts$ is $\le$ the time when the store that writes $v$ to memory is executed in $ooo[i].rob$ in \robDepccm{}, and $mts$ is equal to the time right after that store writes $ccm.m$ in \robDepccm{}.
		\item For each invalidation buffer entry $\langle a, v, [ts_L, ts_U]\rangle$ of $ps[i].ib$ in WMM, $ts_U$ is the time right before value $v$ is overwritten by another store in $ccm.m$.
		If the store that writes $v$ to memory is from processor $i$, then $ts_L$ is $\le$ the time when that store is executed in $ooo[i].rob$ by the \robStExRule{} rule.
		Otherwise, $ts_L$ is equal to the time right after $v$ is written to $ccm.m$.
	\end{enumerate}
	The above invariants can be proved inductively in the same way used in the proof of Theorem 1.
	
	In particular, we consider the case that \robDepccm{} commits a load $L$ to address $a$, which reads the value of store $S$, from $ooo[i].rob$.
	In this case, WMM executes $L$ on $ps[i]$ according to the status of $S$ in \robDepccm{} when $L$ is committed from $ooo[i].rob$:
	\begin{itemize}
		\item If $S$ is still in $ooo[i].sb$ when $L$ commits from ROB, then \wmmDep{} can execute $L$ by reading $ps[i].sb$ (\ie{} the \wmmDepLdMemRule{} rule).
		
		\item If $S$ is in $ccm.m$ at that time, then \wmmDep{} executes $L$ by reading $m$ (\ie{} the \wmmDepLdMemRule{} rule).
		Note that $S$ may be from processor $i$ or another processor $j$.
		In either case, the \wmmDepLdMemRule{} rule maintains all the invariants.
		
		\item The final case is that $S$ has been overwritten by another store in $ccm.m$ before $L$ is committed from ROB.
		In this case, \wmmDep{} executes $L$ by reading $ps[i].ib$ (\ie{} the \wmmDepLdIbRule{} rule).
		For the same reason used in the proof of Theorem 1, the value of $S$ will be in $ps[i].ib$ when $L$ is executed in \wmmDep{}.
		We assume the time interval of $S$ in $ps[i].ib$ is $[ts_L, ts_U]$.
		The guard of the \wmmDepLdIbRule{} rule, \ie{} $ats < ts_U$ ($ats$ is the timestamp of the load address), will be satisfied, because timestamps in WMM is always $\le$ the corresponding time in \robDepccm, and the source register values for the load address must have been computed before the load value is overwritten in $ccm.m$ in \robDepccm.
		No matter $S$ is from processor $i$ or another processor, the way of setting $ts_L$ when $S$ is inserted into $ps[i].ib$ ensures that the timestamp computed for the load result in the \wmmDepLdIbRule{} rule conforms to all the invariants.
	\end{itemize}
\end{proof}

\begin{thm}
	\wmmDep{} $\subseteq$ \robDepccm.
\end{thm}
\begin{proof}
	We still introduce $gts$ into \robDepccm{}, and $gts$ is incremented by one whenever a store writes $ccm.m$.
	Since multiple rules of \robDepccm{} may fire under the same $gts$, we introduce a pair $\langle ut, lt \rangle$ to specify the exact time when each rule fires;
	$ut$ is the upper time, which specifies the value of $gts$ when the rule fires;
	$lt$ is the lower time, a rational number inside $(0,1]$, which is used to order rules with the same upper time.
	When comparing two time pairs, we first compare the $ut$ part, and only compare the $lt$ part when $ut$ parts are equal.
	All \robDeqSbRule{} rules must have $lt=1$, while all other rules in \robDepccm{} must have $0 < lt < 1$.
	
	For any WMM execution $E$, we could construct a rule sequence $E'$ in \robDepccm{}, which has the same program behavior.
	In the construction of $E'$, we always fire the \robDeqSbRule{} and \ccmStRule{} rules atomically to write a store to $ccm.m$, so we only use \robDeqSbRule{} to denote this sequence in the rest of the proof.
	Similarly, we always fire \robLdAddrRule{} and \robLdBypassRule{} atomically to forward data to a load, and always fire \robLdAddrRule{}, \robLdReqRule{}, and \ccmLdRule{} rules atomically to satisfy a load from $ccm.m$, so we will only mention the \robLdBypassRule{} and \robLdReqRule{} rules to refer to the above two atomic sequences in the rest of the proof.
	
	The first part of $E'$ is to fetch all instructions executed in $E$ into the ROB of each processor (using \robFetchRule{} rules).
	The construction of the rest of $E'$ is similar to that in the proof of Theorem 2 (\ie{} WMM $\subseteq$ \robccm{} in the paper), \ie{} when \wmmDep{} writes a store to $m$ or executes an instruction, we write that store to $ccm.m$ or schedule rules to execute and commit that instruction in \robDepccm.
	We maintain the following invariants after each step of construction (the states of \robDepccm{} refers to the states after firing all rules in the constructed $E'$):
	\begin{enumerate}
		\item The states of $m$ and all $sb$ in WMM are the same as the states of $ccm.m$ and all $sb$ in \robDepccm{}.
		\item The $gts$ in WMM is the same as that in \robDepccm.
		\item The upper time assigned to each rule in $E'$ is equal to the value of $gts$ when the rule fires.
		The lower time assigned to each \robDeqSbRule{} is 1, while the lower time assigned to each other rule is within $(0,1)$.
		\item All instructions executed in WMM have also been executed (with the same results) and committed from ROBs in \robDepccm{}.
		\item The value of $gts$ when each instruction is executed or each store is written into $m$ in \wmmDep{} is equal to the upper time of the rule to commit that instruction from ROB or write that store to $ccm.m$ in \robDepccm.
		\item For each $\NmInst$, $\StInst$ or $\LdInst$ instruction executed in \wmmDep, the timestamp computed for the execution result (\ie{} $\NmInst$ instruction result, store address and data, or load result) in \wmmDep{} is equal to the upper time of corresponding rule in \robDepccm{} (\ie{} \robNmExRule, \robStExRule, \robLdBypassRule, or \robLdReqRule) that executes the instruction.
		\item No flush or load re-execution happens in any ROB.
	\end{enumerate}
	Besides proving the above invariants, we also need to show that the rules scheduled in $E'$ can indeed fire, \eg{}, the time of a rule to execute an instruction is smaller than the time of committing that instruction, but is larger than the time of each rule that computes the source operand of that instruction.
	
	The detailed way of constructing $E'$ to simulate each rule in WMM is shown below (the current states of \robDepccm{} refers to the states after firing all existing rules in the constructed $E'$):
	\begin{itemize}
		\item When \wmmDep{} writes a store from $ps[i].sb$ into $m$, we fire the \robDeqSbRule{} rule to write that store from $ooo[i].sb$ to $ccm.m$ at time $\langle ut, 1 \rangle$ in $E'$, in which $ut$ is the current $gts$ in \robDepccm.
		\item When an instruction $I$ is executed by $ps[i]$ in \wmmDep, we commit $I$ from $ooo[i].rob$ in $E'$ at time $\langle ut_c, lt_c \rangle$, in which $ut_c$ is the current $gts$ in \robDepccm, and $lt_c$ is chosen so that this commit rule happens after all existing rules in $E'$.
		If $I$ is $\NmInst$, $\StInst$ or $\LdInst$, we also need to schedule rules to execute it.
		Assume the maximum time among all the rules to compute the source register values of $I$ is $\langle ut_a, lt_a\rangle$.
		(All such rules must have been scheduled in the previous construction steps).
		\begin{itemize}
			\item If $I$ is $\NmInst$ or $\StInst$, then we schedule the corresponding \robNmExRule{} or \robStExRule{} rule for $I$ to fire at time $\langle ut_a, lt_a'\rangle$.
			$lt_a'$ is chosen so that $lt_a < lt_a' < 1$ and $\langle ut_a, lt_a'\rangle < \langle ut_c, lt_c\rangle$.
			\item If $I$ is $\LdInst$, we assume that $I$ reads from a store $S$ in $E$, and that $ooo[i]$ commits the last $\RecInst$ older than $I$ at time $\langle ut_r, lt_r \rangle$ in the previously constructed $E'$.
			\begin{itemize}
				\item If $S$ and $I$ are from the same processor (\ie{} $ooo[i]$), let $\langle ut_s, lt_s \rangle$ be the time of the \robStExRule{} rule for $S$ in the previously constructed $E'$.
				We fire either a \robLdBypassRule{} rule or a \robLdReqRule{} rule (depending on where $S$ is at the rule firing time) to execute $I$ in $E'$ at time $\langle ut_e, lt_e\rangle$, in which $ut_e = \max(ut_a, ut_r, ut_s)$.
				$lt_e$ is chosen so that $\langle ut_c, lt_c\rangle > \langle ut_e, lt_e\rangle > \max(\langle ut_a, lt_a\rangle, \langle ut_r, lt_r \rangle, \langle ut_s, lt_s \rangle)$ and $0 < lt_e < 1$.
				
				\item Otherwise, $S$ and $I$ are from different processors, and let $\langle ut_w, 1 \rangle$ be the time of the \robDeqSbRule{} rule that writes $S$ into $ccm.m$ in the previously constructed $E'$.
				We fire a \robLdReqRule{} rule to execute $I$ at time $\langle ut_e, lt_e\rangle$, in which $ut_e = \max(ut_a, ut_r, ut_w+1)$.
				$lt_e$ is chosen so that $\langle ut_c, lt_c\rangle > \langle ut_e, lt_e\rangle > \max(\langle ut_a, lt_a\rangle, \langle ut_r, lt_r \rangle)$ and $0 < lt_e < 1$.
			\end{itemize}
			Note that $lt_e$ always exists, because $\langle ut_c, lt_c\rangle$ is larger than the time of any existing rule in $E'$.
		\end{itemize}
	\end{itemize}
	Similar to the proof of Theorem 2, the construction of $E'$ here is not in order, but the $E'$ constructed after every step is always a valid rule sequence in \robDepccm{} for all instructions already executed by \wmmDep.
	For the same reason in the proof of Theorem 2, when we schedule rules for instruction $I$ in a construction step, the rules for $I$ will neither affect any existing rule in $E'$ nor depend on any rule scheduled in future construction steps.
	We can prove inductively that the invariants hold and the scheduled rules in $E'$ can indeed fire after each step of construction.
	
	The only non-trivial case is when $ps[i]$ in \wmmDep{} executes a $\LdInst\ a$ instruction $I$, which reads the value of a store $S$ from $ps[j]$ ($j$ may or may not be equal to $i$).
	(In the following proof we will directly use the time variables, \eg{} $ut_a, lt_a, ut_e, lt_e$, \etc, in the above description of the construction step for $I$).
	Due to the way of choosing $\langle ut_e, lt_e \rangle$, at time $\langle ut_e, lt_e \rangle$ in \robDepccm, we are able to compute the load address for $I$, and all $\RecInst$ fences older than $I$ have been committed from $ooo[i].rob$.
	Thus, $I$ could be executed in \robDepccm{} at that time.
	Furthermore, if $i$ is equal to $j$, then store address and data of $S$ must have been computed before $\langle ut_e, lt_e\rangle$ in the constructed $E'$ (because $\langle ut_s, lt_s \rangle < \langle ut_e, lt_e \rangle$); otherwise $S$ must have been written into $ccm.m$ before $\langle ut_e, lt_e\rangle$ in the constructed $E'$ (because $\langle ut_w, 1 \rangle < \langle ut_e, lt_e \rangle$).
	That is, $S$ is visible to $ooo[i]$ at time $\langle ut_e, lt_e\rangle$.
	Besides, $ats$ and $ps[i].rts$ in the \wmmDep{} rule to execute $I$ (\ie{} \wmmDepLdIbRule, \wmmDepLdMemRule{} or \wmmDepLdSbRule) are always equal to $ut_a$ and $ut_r$ respectively according to the invariants.
	The rest of the proof depends on how $I$ is executed in \wmmDep{}.
	\begin{itemize}
		\item \textbf{$I$ is executed by reading $ps[i].ib$ (\ie{} the \wmmDepLdIbRule{} rule):}
		
		For the time interval $[ts_L, ts_U]$ of the $ib$ entry read in the \wmmDepLdIbRule{} rule for $I$, $ts_U$ should be the value of $gts$ in \wmmDep{} when $S$ is overwritten by another store in $m$.
		If $i$ is equal to $j$, then $ts_L$ is the computed timestamp of the $S$ in the \wmmDepStRule{} rule, and should be equal to $ut_s$; otherwise $ts_L$ is one plus the value of $gts$ when $S$ writes $m$ in \wmmDep, and should be equal to $ut_w+1$.
		Then the computed timestamp of the load value of $I$ in \wmmDepLdIbRule{} rule must be equal to $ut_e$.
		
		Now we only need to show that $I$ can read the value of $S$ in \robDepccm{} at time $\langle ut_e, lt_e\rangle$, and that $I$ will not be killed or forced to re-execute later.
		Since the value of $S$ is enqueued into $ps[i].ib$ and stays there until $I$ is executed in \wmmDep, we know the following things about $E$:
		\begin{itemize}
			\item $ps[i].sb$ does not contain any store to $a$ older than $I$ ever since $S$ is written into $m$.
			\item There is no $\RecInst$ older than $I$ executed by $ps[i]$ ever since the overwrite of $S$, \ie{} $ps[i].rts \le ts_U$.
			\item Any load older than $I$ in $ps[i]$ must read from either $S$ or some other store which writes $m$ before $S$.
		\end{itemize}
		According to the invariants and the requirement $ats \le ts_U$ of \wmmDepLdIbRule{}, we can deduce the following implications:
		\begin{enumerate}
			\item $S$ is overwritten in $ccm.m$ at time $\langle ts_U, 1\rangle$.
			\item $ut_a = ats \le ts_U$.
			\item $ut_r = ps[i].rts \le ts_U$.
			\item \label{imp: no store} Neither $ooo[i].rob$ nor $ooo[i].sb$ contains any store to $a$ older than $L$ ever since $S$ is written into $ccm.m$.
			In particular, if $i$ is equal to $j$, then there is no store between $I$ and $S$ in processor $i$.
			\item \label{imp: old load} Any load older than $I$ in $ooo[i]$ must read from either $S$ or some other store which writes $ccm.m$ before $S$.
			Thus, if $I$ is killed by an older load $I'$, then $I'$ must read from a store that writes $ccm.m$ before $S$.
		\end{enumerate}
		Thus, at time $\langle ut_e, lt_e\rangle$ in \robDepccm, $S$ has not been overwritten in $ccm.m$, \ie{} $S$ is in $ooo[i].rob$, $ooo[i].sb$ or $ccm.m$ (note that $S$ is visible to $ooo[i]$ at that time).
		We do a case analysis on whether $i$ is equal to $j$ (\ie{} whether $S$ is also from processor $i$):
		\begin{itemize}
			\item In case $i$ is equal to $j$, the above implication \ref{imp: no store} ensures that $I$ could read the value of $S$ from $ooo[i].rob$ or $ooo[i].sb$ or $ccm.m$ at time $\langle ut_e, lt_e \rangle$, and $I$ will not be killed or forced to re-execute by any older store.
			According to the above implication \ref{imp: old load}, if $I$ is killed by an older load $I'$ later, then $I'$ must be older than $S$, because $I'$ reads from a store that writes $ccm.m$ before $S$.
			However in this case, the $\findLdKilledByLdFunc$ method performed by $I'$ will be stopped by $S$, and $I$ will not be killed.
			
			\item In case $j\neq i$, the above implication \ref{imp: no store} guarantees that $I$ could read the value of $S$ from $ccm.m$ at time $\langle ut_e, lt_e\rangle$, and $I$ will not be killed or forced to re-execute by any older store.
			According to the above implications \ref{imp: old load} and \ref{imp: no store}, if $I$ is killed by an older load $I'$, then $I'$ must get its value before $S$ is written into $ccm.m$, because $I'$ reads from a store that writes $ccm.m$ before $S$.
			However in this case, $I'$ will get its value before $I$ does, and hence $I'$ cannot kill $I$.
		\end{itemize}
		
		\item \textbf{$I$ is executed by reading $m$ (\ie{} the \wmmDepLdMemRule{} rule):}
		
		For the monolithic memory location $m[a]=\langle v, \langle j, sts \rangle, mts \rangle$ read in the \wmmDepLdMemRule{} rule for $I$, $v$ is the value of $S$, $sts$ is the timestamp computed in the \wmmDepStRule{} rule for $S$, $mts$ is one plus the time when $S$ is written into $m[a]$.
		If $i$ is equal to $j$, then $vts$ in the \wmmDepLdMemRule{} rule for $I$ is equal to $sts$, which is also equal to $ut_s$; otherwise we have $vts = mts = ut_w+1$.
		Then the computed timestamp of the load value of $I$ in \wmmDepLdMemRule{} rule must be equal to $ut_e$.
		
		Now we only need to show that $I$ can read the value of $S$ in \robDepccm{} at time $\langle ut_e, lt_e\rangle$, and that $I$ will not be killed or forced to re-execute later.
		According to invariants and the fact that $I$ reads the value of $S$ from $m$, we can deduce the following implications:
		\begin{enumerate}
			\item $S$ is in $ccm.m$ when $L$ is committed from ROB, so $S$ must have not been overwritten in $ccm.m$ at time $\langle ut_e, lt_e \rangle$ ($< \langle ut_c, lt_c \rangle$), \ie{} $S$ is in $ooo[i].rob$, $ooo[i].sb$ or $ccm.m$ at that time (note that $S$ is visible to $ooo[i]$ at that time).
			\item If $i$ is equal to $j$, then there is no store to $a$ between $I$ and $S$ in processor $i$; otherwise there is no store to $a$ older than $I$ in $ooo[i].rob$ or $ooo[i].sb$ ever since $S$ is written into $ccm.m$.
			\item Any load older than $I$ in $ooo[i]$ must read from either $S$ or some other store which writes $ccm.m$ before $S$.
		\end{enumerate}
		The first and second implications ensure that $I$ can read the value of $S$ at time $\langle ut_e, lt_e\rangle$, and $I$ will not be killed or forced to re-execute by any older store.
		The second and third implications ensure that $I$ will not be killed by any older load, for the same reason used in the previous case where $I$ reads from $ps[i].ib$.
		
		\item \textbf{$I$ is executed by reading $ps[i].sb$ (\ie{} the \wmmDepLdSbRule{} rule):}
		
		In this case, $i$ is equal to $j$, \ie{} $S$ is also from processor $i$.
		For the timestamp $sts$ of the store buffer entry read in the \wmmDepLdSbRule{} rule for $I$, $sts$ is the timestamp computed in the \wmmDepStRule{} rule for $S$, and it is equal to $ut_s$ according to invariants.
		Then the computed timestamp of the load value of $I$ in \wmmDepLdMemRule{} rule must be equal to $ut_e$.
		
		Now we only need to show that $I$ can read the value of $S$ in \robDepccm{} at time $\langle ut_e, lt_e\rangle$, and that $I$ will not be killed or forced to re-execute later.
		According to invariants and the fact that $I$ reads the value of $S$ from $ps[i].sb$, we can deduce the following implications:
		\begin{enumerate}
			\item $S$ is in either $ooo[i].rob$ or $ooo[i].sb$ at time $\langle ut_e, lt_e\rangle$ (note that the store address and data of $S$ have been computed at this time).
			\item There is no store to $a$ between $I$ and $S$ in $ooo[i]$.
			\item Any load to $a$ between $I$ and $S$ in $ooo[i]$ must also get the value of $S$ as its result.
		\end{enumerate}
		The first two implications ensure that  $I$ can read the value of $S$ at time $\langle ut_e, lt_e\rangle$, and $I$ will not be killed or forced to re-execute by any older store.
		The last implication ensures that $I$ will not be killed by any older load.
	\end{itemize}
\end{proof}

\section{Proof of \wmmSSB{} Abstracting \robSSBflow}\label{sec:sup-wmm-s}
\begin{thm}
	\robSSBflow{} $\subseteq$ \wmmSSB.
\end{thm}
\begin{proof}
	We use $ooo[i]$ to denote processor $i$ in \robSSBflow, and use $\fm.m$ to denote the monolithic memory in \flowModelName.
	Section 7.4 in the paper has already stated the invariants and the way to simulate the behavior of \robSSBflow{} in \wmmSSB.
	That section has also proved the correctness in case that a load $L$ to address $a$, which reads from a store $S$, is committed from $ooo[i].rob$ in \robSSBflow{}, and that $S$ remains in a segment of \flowModelName{} at the commit time of $L$.
	Here we complete the proof for the remaining cases:
	\begin{itemize}
		\item In case a store $S$ to address $a$ is flowed into $\fm.m$ in the \flowFlowRule{} rule, $S$ must be closer to the root (\ie{} $\fm.m$) than any other stores in \flowModelName.
		According to invariants, in \wmmSSB, all the copies of $S$ must be the the oldest stores to $a$ in their respective store buffers.
		Thus, \wmmSSB{} can fire the \wmmSSBDeqSbRule{} rule to write $S$ into the monolithic memory ($m$), and all the invariants still hold.
		
		\item In case a $\NmInst$, $\RecInst$, or $\StInst$ instruction is committed from $ooo[i].rob$, it is trivial to prove that \wmmSSB{} can fire a \wmmNmRule{}, \wmmRecRule, or \wmmSSBStRule{} rule to execute that instruction, and all the invariants still hold.
		
		\item In case a $\ComInst$ fence is committed from $ooo[i].rob$, the $\ComInst$ response from \flowModelName{} ensures that there must not be any store in any segment of \flowModelName{} which are observed by commits of $ooo[i]$ at that time.
		Therefore, $ps[i].sb$ in \wmmSSB{} should also be empty, and the $\ComInst$ fence can be executed.
		
		\item Consider the case that a load $L$ to address $a$, which reads from a store $S$, is committed from $ooo[i].rob$ in \robSSBflow, and that $S$ is in $\fm.m$ at the commit time of $L$.
		In this case, there cannot be any store to $a$ in any segment of $\fm$ which are observed by commits of $ooo[i]$, because otherwise either the CoRR or CoWR axiom will be violated in \robSSBflow.
		Thus, $ps[i].sb$ cannot have any store to $a$ and $L$ can read the value of $S$ from $m$ in \wmmSSB.
		
		\item Consider the case that a load $L$ to address $a$, which reads from a store $S$, is committed from $ooo[i].rob$ in \robSSBflow, and that $S$ has been overwritten in $\fm.m$ before the commit of $L$.
		In this case,  there cannot be any store to $a$ in any segment of $\fm$ which are observed by commits of $ooo[i]$ right before the overwrite of $S$, because otherwise either the CoRR or CoWR axiom will be violated in \robSSBflow.
		Thus, $ps[i].sb$ cannot have any store to $a$ right before $S$ is overwritten in $m$ in \wmmSSB, so the value of $S$ will be inserted into $ps[i].ib$.
		According to the CoRR and CoWR axioms of \robSSBflow, during the period between the overwrite and the commit of $L$, no store to $a$ can be committed from $ooo[i].rob$, and no load to $a$ committed from $ooo[i].rob$ can read from a store which writes $\fm.m$ after $S$.
		In addition, no $\RecInst$ fence can be committed during that period, otherwise the $\RecInst$ fence will forbid $L$ from reading across it to get the value of $S$. 
		Thus, the value of $S$ will stay in $ps[i].ib$ until $L$ is executed by the \wmmLdIbRule{} rule in \wmmSSB.
	\end{itemize}
\end{proof}

\section{CoRR and CoWR Axioms for Physical Models}\label{sec:sup-sc}

In the paper, we have introduced the following axioms
($L_1,L_2,S_1,S_2$ denote loads and stores to the same address):
\begin{description}
	\item[CoRR] (Read-Read Coherence): $L_1\poEdge L_2\ \wedge\ S_1\rfEdge L_1\ \wedge\ S_2\rfEdge L_2 \Longrightarrow S_1 == S_2\ \vee\ S_1\coEdge S_2$. 
	\item[CoWR] (Write-Read Coherence): $S_2\rfEdge L_1\ \wedge\ S_1\poEdge L_1 \Longrightarrow S_1 == S2\ \vee\ S_1\coEdge S_2$..
\end{description}
We have used the fact that physical models (\ie{} \robccm, \robDepccm{} and \robSSBflow) satisfy these two axioms in the proofs of the relations between \IIE{} models with physical models.
Now we formally prove that these axioms hold for all the physical models.
We will directly use $L_1,L_2,S_1,S_2$ in the proofs, and all the operations about $L_1$ and $L_2$ discussed in the proofs are the final operations of $L_1$ and $L_2$ to get their load results, \ie{} $L_1$ and $L_2$ should not be killed or forced to re-execute afterwards.

\subsection{CoRR and CoWR Axioms for \ccmModelName+\robModelName/\robDepModelName}

\begin{lemma}
	\robccm{} satisfies the CoRR axiom.
\end{lemma}
\begin{proof}
	We assume $L_1$ and $L_2$ are both from processor $i$ ($ooo[i]$).
	We do a case analysis on how $L_1$ get its final result, \ie{} the value of $S_1$, and prove that $S_2\coEdge S_1$ is impossible in each case.
	
	First consider the case that $L_1$ gets its final result via the \robLdBypassRule{} rule.
	In this case, $S_1$ is also from $ooo[i]$.
	If we have $S_2\coEdge S_1$, then $L_2$ must get its final result, \ie{} the value of $S_2$, before the \robStExRule{} rule for $S_1$ has fired.
	However, in this case, $L_2$ will be killed later when the \robStExRule{} rule for $S_1$ fires.
	
	Next consider the case that $L_1$ reads the value of $S_1$ from the monolithic memory of \ccmModelName{} ($ccm.m$).
	If we have $S_2\coEdge S_1$, then $L_2$ must get its final result before $L_1$ gets the response from \ccmModelName, and there should not be any store to $a$ between $L_1$ and $L_2$ in $ooo[i]$ (otherwise $L_2$ will be killed by the store).
	However, in this case, when $L_1$ gets its response from \ccmModelName, it will kill $L_2$.
\end{proof}

\begin{lemma}
	\robccm{} satisfies the CoWR axiom.
\end{lemma}
\begin{proof}
	We assume $L_1$ and $S_1$ are both from processor $i$ ($ooo[i]$).
	If we have $S_2\coEdge S_1$, then $L_1$ must get its final result, \ie{} the value of $S_2$, before the \robStExRule{} rule for $S_1$ has fired.
	However, in this case, $L_1$ will be killed later when the \robStExRule{} rule for $S_1$ fires.
\end{proof}

Since \robDepccm{} $\subseteq$ \robccm, \robDepccm{} also satisfies the CoRR and CoWR axioms.

\subsection{CoRR and CoWR Axioms for \robSSBflow}

\begin{lemma}
	\robSSBflow{} satisfies the CoRR axiom.
\end{lemma}
\begin{proof}
	We assume $L_1$ and $L_2$ are both from processor $i$ ($ooo[i]$).
	We do a case analysis on how $L_1$ get its final result, \ie{} the value of $S_1$, and prove that $S_2\coEdge S_1$ is impossible in each case.
	
	First consider the case that $L_1$ gets its final result via the \robLdBypassRule{} rule.
	In this case, $S_1$ is also from $ooo[i]$.
	If we have $S_2\coEdge S_1$, then $L_2$ must get its final result, \ie{} the value of $S_2$, before the \robStExRule{} rule for $S_1$ has fired.
	However, in this case, $L_2$ will be killed later when the \robStExRule{} rule for $S_1$ fires.
	
	Next consider the case that $L_1$ reads the value of $S_1$ from \flowModelName.
	If we have $S_2\coEdge S_1$, then before $L_1$ issues its request to \flowModelName, $L_2$ must either issue its request to \flowModelName{} or get bypassing from ROB .
	Furthermore, there should not be any store to $a$ between $L_1$ and $L_2$ in $ooo[i]$ (otherwise $L_2$ will be killed by the store).
	However, in this case, when $L_1$ issues its request to \flowModelName, it will kill $L_2$.
\end{proof}

\begin{lemma}
	\robSSBflow{} satisfies the CoWR axiom.
\end{lemma}
\begin{proof}
	The argument is the same as that for \robccm.
\end{proof}

\end{document}